\documentclass{scrartcl}

\usepackage{graphicx}        
\usepackage{multirow}
\usepackage{enumitem} 
\usepackage{framed} 

\usepackage{amsbsy}
\usepackage{amsfonts, amsthm}
\usepackage{amsmath,amssymb}
\usepackage{natbib}
\usepackage{eso-pic} 
\usepackage{hyperref} 
\providecommand{\bo}{\mathbf}
\providecommand{\bs}{\boldsymbol}
\providecommand{\cov}{\mathrm{\bo Cov}}
\DeclareMathOperator*{\argmax}{arg\,max\,}
\DeclareMathOperator*{\argmin}{arg\,min\,}
\DeclareMathOperator{\tr}{tr}
\DeclareMathOperator{\diag}{diag}
\DeclareMathOperator{\var}{\mathrm{\bo Var}}

\DeclareMathOperator{\Sign}{Sign}
\DeclareMathOperator{\sign}{sign}
\DeclareMathOperator{\scale}{scale}
\DeclareMathOperator{\med}{med}

\theoremstyle{plain}
\newtheorem{theorem}{Theorem}[section]
\newtheorem{proposition}[theorem]{Proposition}
\newtheorem{lemma}[theorem]{Lemma}

\begin{document}

\AddToShipoutPictureBG*{%
  \AtPageUpperLeft{%
    \raisebox{-\baselineskip}{%
      \makebox[5pt][l]{\quad \small This is an Author's Accepted Manuscript of a book chapter to be published in Nordhausen, K., Taskinen, S. (Eds.) \textit{Modern}}}
\raisebox{-2\baselineskip}{\makebox[5pt][l]{
	\textit{Multivariate and Robust Methods,} Springer, 2015}}
}}%

\title{Robust high-dimensional precision matrix estimation}
\author{Viktoria \"Ollerer and Christophe Croux}
\date{\small{Faculty of Economics and Business, KU Leuven, 3000 Leuven, Belgium\\E-mail: viktoria.oellerer@kuleuven.be, christophe.croux@kuleuven.be}}
\maketitle

\abstract{The dependency structure of multivariate data can be analyzed using the covariance matrix $\bs \Sigma$. In many fields the precision matrix $\bs \Sigma^{-1}$ is even more informative. As the sample covariance estimator is singular in high-dimensions, it cannot be used to obtain a precision matrix estimator. A popular high-dimensional estimator is the graphical lasso, but it lacks robustness. We consider the high-dimensional independent contamination model. Here, even a small percentage of contaminated cells in the data matrix may lead to a high percentage of contaminated rows. Downweighting entire observations, which is done by traditional robust procedures, would then results in a loss of information. In this paper, we formally prove that replacing the sample covariance matrix in the graphical lasso with an elementwise robust covariance matrix leads to an elementwise robust, sparse precision matrix estimator computable in high-dimensions. Examples of such elementwise robust covariance estimators are given. The final precision matrix estimator is positive definite, has a high breakdown point under elementwise contamination and can be computed fast.}

\section{Introduction}
\label{Oellerer:sec:Intro}

Many statistical methods that deal with the dependence structures of multivariate data sets start from an estimate of the covariance matrix. For observations $\mathbf{x}_1,\ldots, \mathbf{x}_n\in\mathbb{R}^p$ with $n>p$, the classical sample covariance matrix \begin{align}
\label{Oellerer:eq:sample}
\hat{\bs \Sigma}=\frac{1}{n-1}\sum_{i=1}^n(\mathbf{x}_{i}-\bar{\mathbf{x}})(\mathbf{x}_{i}-\bar{\mathbf{x}})^\top,
\end{align}
where $\bar{\mathbf{x}}\in\mathbb{R}^p$ denotes the mean of the data, is optimal in many ways. It is easy to compute, maximizes the likelihood function for normal data, is unbiased and consistent. However, problems arise when $p$ increases. For $p\approx n$, the sample covariance matrix has low precision and for $p>n$ it even becomes singular, such that the estimated precision matrix $\hat{\bs\Theta}:=\hat{\bs \Sigma}^{-1}$ is not computable anymore. This is a problem since there are many fields where the precision matrix is needed rather than the covariance matrix. Computation of Mahalanobis distances or linear discriminant analysis are just two examples. The most popular field using precision matrices is probably Gaussian graphical modeling, where the nodes of the graph represent the different variables. If an element $(\hat{\bs \Theta})_{ij}$ of the estimated precision matrix equals zero, the variables $i$ and $j$ are independent given all the other variables, and no edge is drawn between the nodes representing variables $i$ and $j$. Therefore, edges correspond to nonzero elements of the precision matrix. As a result, the whole graph can be recovered if the support of the precision matrix is known. This leads to an increasing interest in sparse precision matrices (precision matrices with a lot of zero elements) as interpretation of the graph will be eased if the number of nonzeros in the precision matrix is kept small. 

The three most suitable techniques to compute sparse precision matrices that are also applicable in high dimensions are the graphical lasso (GLASSO) \citep{Friedman2}, the quadratic approximation method for sparse inverse covariance learning (QUIC) \citep{Hsieh} and the constrained $L_1$-minimization for inverse matrix estimation (CLIME) \citep{Cai2}. All three methods start from the sample covariance matrix $\hat{\bs \Sigma}$ and try to minimize a criterion based on the log-likelihood (see Section~\ref{Oellerer:sec:Glasso}). Since these estimators use the nonrobust sample covariance matrix as an input, they are only suitable for clean data that do not contain any outliers.

The problem, however, is that data is rarely clean. Thus, there is need for robust procedures. Most robust procedures downweight observations as a whole (`rowwise downweighting'). However, in many statistical applications only a limited number of observations are available, while large amounts of variables are measured for each observation. Downweighting an entire observation because of one single outlying cell in the data matrix results in a huge loss of information. Additionally, the contaminating mechanism may be independent for different variables. In this case, the probability of having an observation without contamination in any cell is decreasing exponentially when the number of variables increases. As an example, imagine a data set, where each observation contains exactly one contaminated cell. Even though there is not a single fully clean observation, each observation still contains a lot of clean information. Nonetheless, the `classical' robust procedures (that downweight whole observations) cannot deal with a data set like that, since they need at least half of the observations to be absolutely clean of contamination. This type of `cellwise' or `elementwise' contamination was formally described by \cite{Alqallaf}, who extend the usual Tukey-Huber contamination model (the model that considers whole observations as either outlying or clean). In this more extensive setup, a random vector 
\begin{align*}
\mathbf{x}=(\bo I_p-\bo B)\mathbf{y}+\bo B\mathbf{z}
\end{align*}
is observed, where $\mathbf{y}$ follows the model distribution and $\mathbf{z}$ some arbitrary distribution creating contamination, and $\mathbf{y},\bo B$ and $\mathbf{z}$ are independent. Depending on the Bernoulli random variables $B_i$ with $\mathbb{P}[B_i=1]=\epsilon_i$ that build the diagonal matrix $\bo B=\diag(B_1,\ldots,B_p)$, different types of outliers are created. If all $B_i$ are independent ($i=1,\ldots,p$), we speak about `cellwise contamination'. If $\mathbb{P}[B_1=B_2=\ldots=B_p]=1$, rowwise contamination is created. Under any type of contamination, the sample covariance matrix $\hat{\bs \Sigma}$ is not a good estimator anymore, as it can be distorted by just a single outlying observation.

For robust covariance matrix estimation under rowwise contamination, a lot of work has been done. One of the most popular rowwise robust covariance estimators is the minimum covariance determinant \citep{RousseeuwMCD}. It has a high breakdown point and is very fast to compute. However, it is not computable in high-dimensions. Another estimator with very nice theoretical properties is the affine equivariant rank covariance matrix \citep{Oja2}. It is very efficient and has maximal breakdown point. However, its computation is extremely time consuming, especially in high-dimensions. \cite{Maronna2} propose a high-dimensional covariance estimator, an orthogonalized version of the Gnanadesikan-Kettenring estimate (OGK). Another very simple estimator has been developed by \cite{Oja}. Their spatial sign covariance matrix appeals through a simple definition and can be computed very fast, even in high-dimensions. Very recently, \cite{Ollila} have introduced a regularized $M$-estimator of scatter. Under general conditions, they show existence and uniqueness of the estimator, using the concept of geodesic convexity. 

Much less work has been done for covariance estimation under cellwise contamination. A first approach was taken by \cite{VanAelst2}, who defined a cellwise weighting scheme for the Stahel-Donoho estimator. However, as for the original estimate, computation times are not feasible for larger numbers of variables. A very recent approach by \cite{Agostinelli} flags cellwise outliers as missing values and applies afterwards a rowwise robust method that can deal with missing values. By this, it can deal with cellwise and rowwise outliers at the same time, but again, computation for high-dimensions is not achievable. 

The first step to deal with cellwise outliers in very high-dimensions has been taken by \cite{Alqallaf2}. They first compute a pairwise correlation matrix. Afterwards the OGK estimate is applied to obtain a positive semidefinite covariance estimate. This method has been fine tuned by \cite{Tarr2} who use pairwise covariances instead of correlations \citep[see also][]{Garth}. This matrix is then plugged into the graphical lasso (and similar techniques) instead of the sample covariance matrix, resulting in a sparse precision matrix estimate. A very different approach has been taken by \cite{Finegold}. Replacing the assumption of Gaussian distribution of the data with t-distribution gives more robust results since the t-distribution has heavier tails. Assuming a so-called `alternative' t-distribution (see Section~\ref{Oellerer:sec:sim}) results in robustness against cellwise contamination.

In this paper, we consider different high-dimensional precision matrix estimators robust to cellwise contamination. Our approach is similar in spirit as \cite{Tarr2} \citep[see also][]{Garth}, but we emphasize the difference in Section~\ref{Oellerer:sec:def}. We start with pairwise robust correlation estimates from which we then estimate a covariance matrix by multiplication with robust standard deviations. This cellwise robust covariance matrix replaces then the sample covariance matrix in the GLASSO, yielding a sparse, cellwise robust precision matrix estimator. The different nonrobust precision matrix estimators are introduced in Section~\ref{Oellerer:sec:Glasso}. The cellwise robust covariance matrix estimators are explained in Section~\ref{Oellerer:sec:def}. We discuss the selection of the regularization parameter in Section~\ref{Oellerer:sec:rho}. In Section~\ref{Oellerer:sec:BP}, the breakdown point of the proposed precision matrix estimator is derived. Simulation studies are presented in Section~\ref{Oellerer:sec:sim}. In Section~\ref{Oellerer:sec:app}, we discuss possible applications of the proposed method and present a real data example. Section~\ref{Oellerer:sec:conclusion} concludes.

\section{High-dimensional sparse precision matrix estimation for clean data}
\label{Oellerer:sec:Glasso}

Recently, a lot of effort has been put into designing estimators and efficient routines for high-dimensional precision matrix estimation. We focus here on sparse precision matrix estimation, that is, procedures that result in a precision matrix containing many zero elements. In this section, we review three techniques that start from an estimate of the covariance matrix $\hat{\bs \Sigma}$ and then optimize a criterion based on the likelihood function to find the precision matrix estimate. Since the methods are based on the sample covariance matrix, they are only useful if no contamination is present in the data.

The graphical lasso (GLASSO) \citep{Friedman2} maximizes the $L_1$-penalized log-likelihood function:
\begin{align}
\label{Oellerer:eq:GLASSO}
\hat{\bs \Theta}_{GL}(\bo X)=\argmax_{\substack{\bs \Theta\in\mathbb{R}^{p\times p}\\ \bs \Theta\succ0}} \log \det(\bs \Theta)-\tr(\hat{\bs \Sigma}\bs \Theta)-\rho\sum_{j,k=1}^p|(\Theta)_{jk}|,
\end{align}
where $\bo A \succ 0$ denotes a strictly positive definite matrix $\bo A$ and $\rho$ is a regularization parameter. If the regularization parameter $\rho$ is equal to zero, the solution of the GLASSO is the inverse of the sample covariance matrix. The larger the value of $\rho$ is chosen, the more sparse the precision matrix estimate becomes. Since the objective function (\ref{Oellerer:eq:GLASSO}) is concave, there exists a unique solution. \cite{Banerjee} showed that the solution of the GLASSO always results in a strictly positive definite estimate $\hat{\bs \Theta}_{GL}(\bo X)$ for any $\rho>0$, even if $p>n$, and this for any positive semidefinite, symmetric matrix $\hat{\bs \Sigma}$ in (\ref{Oellerer:eq:GLASSO}).

The solution $\hat{\bs \Theta}_{GL}(\bo X)$ can be computed via iterative multiple lasso regression in a block coordinate descent fashion. That means that each column of the final estimate is computed separately. Looking at the first order condition only for the target column, the equation can be seen as a first order condition of a multiple lasso regression. The GLASSO algorithm loops through all columns of the precision matrix iteratively, computing each time the multiple lasso regression, until convergence of the precision matrix estimate is reached. Note that the algorithm does not use the data directly, but only uses it indirectly by using the sample covariance matrix. The GLASSO algorithm is implemented in \textsf{Fortran} and available through the \textsf{R}-package \texttt{glasso} \citep{glasso}. However, this implementation sometimes encounters convergence problems. Therefore, we use in the remainder of this paper, the implementation of the GLASSO algorithm in the \textsf{R}-package \texttt{huge} \citep{huge}, where these convergence issues have been solved.

Another algorithm solving (\ref{Oellerer:eq:GLASSO}) is the quadratic approximation method for sparse inverse covariance learning (QUIC) \citep{Hsieh}. Both the GLASSO algorithm and the QUIC compute a solution for the same objective function. It turned out that QUIC was performing considerably slower in high dimensions than the GLASSO implementation in the R-package huge \citep{huge}, and therefore we will not deal with the former in this paper. 

\medskip
The constrained $L_1$-minimization for inverse matrix estimation (CLIME) is defined as
\begin{gather*}
\hat{\bs \Theta}_1(\bo X)=\argmin_{\bs \Theta\in\mathbb{R}^{p\times p}} \sum_{i=1}^p\sum_{j=1}^p |\bs \Theta_{ij}|\qquad \text{ subject to } \max_{\substack{i=1,\ldots,p\\j=1,\ldots,p}} |(\hat{\bs \Sigma}\bs \Theta-\bo I_p)_{ij}|\leq \rho\\
\hat{\bs \Theta}_C(\bo X)=(\hat{\theta}_{ij}) \quad \text{ with } \quad\hat{\theta}_{ij}=\hat{\theta}_{ij}^1I_{[|\hat{\theta}_{ij}^1|\leq|\hat{\theta}_{ji}^1|]}+\hat{\theta}_{ji}^1I_{[|\hat{\theta}_{ij}^1|>|\hat{\theta}_{ji}^1|]}\quad\text{ and }\quad \hat{\Theta}_1(\bo X)=(\hat{\theta}_{ij}^1).
\end{gather*}
The result is a symmetric matrix that is positive definite with high probability. The CLIME estimator $\hat{\bs \Theta}_C(\bo X)$ converges fast towards the true precision matrix under some mild conditions. The algorithm is implemented in the \textsf{R}-package \texttt{clime} \citep{clime}. Like the GLASSO algorithm, it does not use the data directly, but ony requires the sample covariance matrix as an input. Replacing the sample covariance matrix with a cellwise robust estimator (see Section~\ref{Oellerer:sec:def}), the resulting estimator is similarly accurate (with respect to Kullback-Leibler divergence measure, see Section~\ref{Oellerer:sec:sim}) as the one obtained when plugging the cellwise robust estimator into the GLASSO estimator. In some cases, plugging the robust estimator into the CLIME led to slightly better accuracy. However, the computation time, was much higher than when plugging it into the GLASSO (for $p=60$ the computation time was more than 10 times higher). Since in high-dimensional analysis computation time is important, we will not consider this estimator in the remainder of the paper.

\section{Cellwise robust, sparse precision matrix estimators}
\label{Oellerer:sec:def}

We start with computing a cellwise robust covariance matrix $\bo S$ by pairwise, robust estimation of the covariances. This cellwise robust covariance matrix can then be used to replace the sample covariance matrix in the GLASSO estimator (or another sparse precision matrix estimator). This results in a sparse, cellwise robust precision matrix estimate. Our approach differs from \cite{Tarr2} in the selection of the initial covariance estimate. We estimate robust correlations and standard deviations separately to get the robust covariances. The resulting covariance matrix is then always positive semidefinite. This leads to a simplification of the estimator, increases the breakdown point and speeds up computation substantially.

\subsection{Robust covariance matrix estimation based on pairwise covariances}
\label{Oellerer:sec:cov}
\cite{Tarr2} use the approach of \cite{Gnanadesikan} to obtain a robust, pairwise covariance estimate. It is based on the idea that the robust covariance of two random variables $X$ and $Y$ can be computed using a robust variance. For the population covariance $\cov$ and the population variance $\var$, the following identity holds
\begin{align}
\label{Oellerer:eq:GK}
\cov(X,Y)=\frac{1}{4\alpha\beta}[\var(\alpha X+\beta Y)-\var(\alpha X-\beta Y)]
\end{align}
with $\alpha=1/\sqrt{\var(X)}$, $\beta=1/\sqrt{\var(Y)}$. If $\var$ is replaced by a robust variance estimator, a robust covariance estimate can be obtained. 

This approach has two drawbacks. Firstly, the addition and subtraction of different variables leads to a propagation of the outliers. Therefore, the resulting estimator has a breakdown point of less than $25$\% under cellwise contamination. Secondly, the resulting covariance matrix is not necessarily positive semidefinite. Therefore, \cite{Tarr2} need to apply methods that `make' the matrix positive  semidefinite to be able to use this covariance matrix estimate as a replacement of the sample covariance matrix in a sparse precision matrix estimator. To this end, they use the orthogonalized Gnanadesikan-Kettenring (OGK) approach \citep{Maronna2} as well as the computation of the nearest positive (semi)definite (NPD) matrix as suggested by \cite{Higham}. Starting from an estimate $\tilde{\bo S}\in\mathbb{R}^{p\times p}$ for the covariance matrix of the data $\bo X\in\mathbb{R}^{n\times p}$, NPD finds the closest positive semidefinite matrix $\bo S$ to the covariance estimate $\tilde{\bo S}$ in terms of the Frobenius norm
\begin{align*}
\bo S=\min_{\hat{\bo S}\succeq 0}\|\tilde{\bo S}-\hat{\bo S}\|_F,
\end{align*}
where $\|\bo A\|_F=\sum_{j,k=1}^pa_{jk}^2$ for a matrix $\bo A=(a_{jk})_{j,k=1,\ldots,p}\in\mathbb{R}^{p\times p}$ and $\bo A\succeq0$ denotes a positive semidefinite matrix. An algorithm to compute the nearest matrix $\bo S$ is implemented in the \textsf{R}-package \texttt{Matrix} under the command \texttt{nearPD()}. 
 In our simulations, we observed that NPD gave in general better results than OGK and could also be computed considerably faster. 

\subsection{Robust covariance matrix estimation based on pairwise correlations}
\label{Oellerer:sec:cor}
In contrast to \cite{Tarr2}, we use a robust correlation estimator $r(\cdot)$ to estimate the pairwise covariance matrix $(s_{jk})\in\mathbb{R}^{p\times p}$
\begin{align}
\label{Oellerer:eq:pairwise}
s_{jk}=\scale(\mathbf{x}^j)\scale(\mathbf{x}^k) r(\mathbf{x}^j,\mathbf{x}^k)\qquad j,k=1,\ldots,p
\end{align}
from the data $\bo X=(\mathbf{x}^1,\ldots, \mathbf{x}^p)\in\mathbb{R}^{n\times p}$, where $\scale()$ is a robust scale estimate like the median absolute deviation or the $Q_n$-estimator \citep{RousseeuwQn}. Both estimators are equally robust with a breakdown point of 50\%. Since the $Q_n$-estimator is more efficient at the Gaussian model and does not need a location estimate, we opt for this scale estimate. The amount of contamination that the resulting covariance matrix $\bo S=(s_{jk})_{j,k=1,\ldots,p}$ can withstand depends then on the breakdown point of the scale estimator used (see Section~\ref{Oellerer:sec:BP}). Using the $Q_n$-scale, we obtain an estimator with a breakdown point of $50$\% under cellwise contamination.

There are different possibilities for choosing a robust correlation estimator. Gaussian rank correlation \citep[e.g.][]{Boudt} is defined as the sample correlation estimated from the Van Der Waerden scores (or normal scores) of the data
\begin{align}
\label{Oellerer:eq:Gauss}
r_{Gauss}(\mathbf{x}^j,\mathbf{x}^k)=\frac{\sum_{i=1}^n \Phi^{-1}(\frac{R(x_{ij})}{n+1})\Phi^{-1}(\frac{R(x_{ik})}{n+1})}{\sum_{i=1}^n(\Phi^{-1}(\frac{i}{n+1}))^2},
\end{align}
where $R(x_{ij})$ denotes the rank of $x_{ij}$ among all elements of $\mathbf{x}^j$, the $j$th column of the data matrix. Similarly  $R(x_{ik})$ stands for the rank of $x_{ik}$ among all elements of $\mathbf{x}^k$. Gaussian rank correlation is robust and consistent at the normal model. Still it is asymptotically equally efficient as the sample correlation coefficient at normal data. This makes it a very appealing robust correlation estimator. Note that the Gaussian rank correlations can easily be computed as the sample covariance matrix from the ranks $R(x_{ij})$ of the data. Since the sample covariance matrix is positive semidefinite, the covariance matrix $\bo S$ using Gaussian rank correlation is also positive semidefinite. Therefore, we do not need to apply NPD or OGK to obtain a positive semidefinite covariance estimate. This saves computation time and simplifies the final precision matrix estimator.

Another robust correlation estimator is Spearman correlation \citep{Spearman}. It is defined as the sample correlation of the ranks of the observations:
\begin{align*}
r_{Spearman}(\mathbf{x}^j,\mathbf{x}^k)=\sum_{i=1}^n\frac{(R(x_{ij})-\frac{n+1}{2})(R(x_{ik})-\frac{n+1}{2})}{\sqrt{\sum_{i=1}^n(R(x_{ij})-\frac{n+1}{2})^2\sum_{i=1}^n(R(x_{ik})-\frac{n+1}{2})^2}}.
\end{align*}
Spearman correlation is slightly less efficient than Gaussian rank correlation. Additionally, it is not consistent at the normal model. To obtain consistency, the correlation estimator needs to be non linearly transformed. The transformation, however, destroys the positive semidefiniteness of the estimator $\bo S$, and therefore we do not apply it. In our opinion, the inconsistency is not a huge problem because the asymptotic bias of the Spearman correlation is at most 0.018 \citep{Boudt}. This is also confirmed by the simulations in Section~\ref{Oellerer:sec:sim}, where similar results are obtained with Spearman correlation as with Gaussian rank correlation.

We also consider Quadrant correlation \citep{Blomqvist}. Quadrant correlation is defined as the frequency of centered observations in the first and third quadrant, minus the frequency of centered observations in the second and forth quadrant
\begin{align*}
r_{Quadrant}(\mathbf{x}^j,\mathbf{x}^k)=\frac{1}{n}\sum_{i=1}^n\sign((x_{ij}-\med_{\ell=1,\ldots,n}x_{\ell j})(x_{ik}-\med_{\ell=1,\ldots,n}x_{\ell k})),
\end{align*}
where $\sign(\cdot)$ denotes the sign-function. Quadrant correlation is less efficient than Gaussian rank correlation and Spearman correlation \citep{Dehon}. Like Spearman correlation, Quadrant correlation is only consistent at the normal model if a transformation is applied to the correlation estimate. The final covariance matrix of the consistent Quadrant correlation is no longer positive semidefinite. Since we need a positive semidefinite covariance matrix, we opt for the inconsistent Quadrant correlation. Note that the asymptotic bias at the normal distribution of the inconsistent Quadrant correlation is substantially higher than for Spearman correlation. Taking all this drawback of Quadrant correlation into account, it is not a surprise that we obtain worse simulation results with Quadrant correlation than with Spearman or Gaussian rank correlation in Section~\ref{Oellerer:sec:sim}.

\subsection{Cellwise robust precision matrix estimation}

To obtain a cellwise robust precision matrix estimator, we adapt the definition of the GLASSO estimator given in (\ref{Oellerer:eq:GLASSO}). Recall that GLASSO takes the sample covariance estimator as an input and returns a sparse estimate of the precision matrix as an output. We will replace the sample covariance matrix by the cellwise robust covariance matrices $\bo S$ of Sections~\ref{Oellerer:sec:cov} and \ref{Oellerer:sec:cor} in order to obtain a cellwise robust, sparse precision matrix estimator. Hence, we obtain the following estimator
\begin{align}
\label{Oellerer:eq:GLASSOrob}
\hat{\bs \Theta}_{\bo S}(\bo X)=\argmax_{\substack{\bs \Theta=(\theta_{jk})\in\mathbb{R}^{p\times p}\\ \bs \Theta\succ0}} \log \det(\bs \Theta)-\tr(\bo S\bs \Theta)-\rho\sum_{j,k=1}^p|\theta_{jk}|,
\end{align}
If $\bo S$ is a robust covariance matrix based on pairwise correlations as in Section~\ref{Oellerer:sec:cor}, we refer to $\hat{\bs \Theta}_{\bo S}(\bo X)$ as `correlation based precision matrix estimator'. If $\bo S$ is a covariance matrix based on pairwise covariances as in Section~\ref{Oellerer:sec:cov}, we call $\hat{\bs\Theta}_{\bo S}(\bo X)$ `covariance based precision matrix estimator'. Since the algorithm for computing the GLASSO only requires a positive semidefinite, symmetric matrix $\bo S$ as an input and not the data, we use it to compute $\hat{\bs \Theta}_{\bo S}(\bo X)$.

Like for the original GLASSO algorithm, the final precision matrix estimate $\hat{\bs \Theta}_{\bo S}(\bo X)$ will always be positive definite as long as the initial covariance matrix $\bo S$ is positive semidefinite, even if $p>n$. Therefore, it is important that the initial covariance estimate $\bo S$ is positive semidefinite.

The final precision matrix estimator $\hat{\bs\Theta}_{\bo S}(\bo X)$ will inherit the breakdown point of the initial covariance matrix $\bo S$ (see Section~\ref{Oellerer:sec:BP}). As a result, the correlation based precision matrix estimator has a breakdown point of 50\% under cellwise contamination, while the covariance estimators based on pairwise covariances can have a breakdown point of at most 25\% under cellwise contamination.

The covariance matrices based on pairwise correlations we considered (i.e. the matrices based on Gaussian correlation, Spearman correlation, and Quadrant correlation) are all positive semidefinite. Indeed, they can be computed as sample correlation matrices of transformed data. For instance, the quadrant correlation matrix is a sample correlation matrix of the signs of the differences of the observations to their median. In contrast, covariance matrices based on pairwise covariances need to be transformed to be positive semidefinite for which we used the NPD method described in Section~\ref{Oellerer:sec:cov}. Additionally, all pairwise robust covariances need to be computed according to (\ref{Oellerer:eq:GK}), which may become very time consuming. Therefore, the correlation based precision matrix estimators are much faster to compute than the covariance based precision matrix estimators.

To sum up, correlation based precision matrix estimators are faster to compute and feature a higher breakdown point under cellwise contamination than covariance based precision matrix estimators.

\section{Selection of the regularization parameter $\rho$}
\label{Oellerer:sec:rho}

When selecting the regularization parameter $\rho$, a good trade-off between a high value of the likelihood function and the sparseness of the final precision matrix has to be found. The two most common methods to find the optimal trade-off are the Bayesian Information Criterion (BIC) and cross validation (CV).

The BIC for a $L_1$-regularized precision matrix estimator $\hat{\bs \Theta}_\rho$ for a fixed value of $\rho$ has been given in \cite{Yuan}:
\begin{align*}
BIC_{classic}(\rho)=-\log\det\hat{\bs\Theta}_\rho+\tr(\hat{\bs \Theta}_\rho\hat{\bs \Sigma})+\frac{\log n}{n}\sum_{i\leq j}\hat{e}_{ij}(\rho)
\end{align*}
with $\hat{\bs \Sigma}$ the sample covariance estimate and $\hat{e}_{ij}=1$ if $(\hat{\bs\Theta}_\rho)_{ij}\neq0$ and $\hat{e}_{ij}=0$ otherwise. To obtain a cellwise robust BIC criterion, we replace $\hat{\bs \Sigma}$ by a cellwise robust covariance matrix $\bo S$ and use a cellwise robust precision matrix estimator $\hat{\bs\Theta}_{\rho}$:
\begin{align*}
BIC(\rho)=-\log\det\hat{\bs\Theta}_\rho+\tr(\hat{\bs \Theta}_\rho\bo S)+\frac{\log n}{n}\sum_{i\leq j}\hat{e}_{ij}(\rho).
\end{align*}
Computing the value of BIC over a grid, the value $\rho$ yielding the lowest BIC is chosen.

To perform $K$-fold cross validation, the data first has to be split into $K$ blocks of nearly equal size $n_k$ ($k=1,\ldots,K$). Each block $k$ is left out once and used as test data $(\mathbf{x}^1_{(k)},\ldots,\mathbf{x}^p_{(k)})$. On the remaining data, the precision matrix estimate $\hat{\bs \Theta}_\rho^{(-k)}$ is computed using the regularization parameter $\rho$. As an evaluation criterion, the negative log-likelihood on the test data is computed
\begin{align*}
L^{(k)}(\rho)=-\log\det\hat{\bs \Theta}_\rho^{(-k)}+\tr(\bo S^{(k)}\hat{\bs\Theta}_\rho^{(-k)}),
\end{align*}
where $\bo S^{(k)}$ is the initial robust covariance estimate computed on the test data, i.e.
\begin{align*}
(\bo S^{(k)})_{ij}=\scale(\mathbf{x}^i_{(k)})\scale(\mathbf{x}^j_{(k)})r(\mathbf{x}^i_{(k)}, \mathbf{x}^j_{(k)}) \qquad i,j=1,\ldots,p
\end{align*}
exactly as in Equation~(\ref{Oellerer:eq:pairwise}). By using a robust covariance estimate computed from the test data, outliers present in the test data will not affect the cross-validation criterion too much. This is done over a range of values of $\rho$. The value of $\rho$ minimizing the negative log-likelihood is chosen as the final regularization parameter
\begin{align}
\label{Oellerer:eq:CV}
\hat{\rho}=\argmin_{\rho}\frac{1}{K}\sum_{k=1}^K L^{(k)}(\rho).
\end{align}

As pointed out by a referee, it could occure that some of the test data sets include a percentage of outliers exceeding the breakdown point of the precision matrix estimator, leading to possible breakdown of the cross validation procedure. In our numerical experiments, with contamination levels low compared to the breakdown point and independent for different cells, we did not face this problem. Replacing the sum in (\ref{Oellerer:eq:CV}) by a median, for instance, may provide a way out.

To select a grid of values of $\rho$, we suggest to use the heuristic approach implemented in the \texttt{huge}-package \citep{huge}. It chooses a logarithmic spaced grid of ten values. The largest value of the grid depends on the value of the initial covariance matrix $\bo S$
\begin{align*}
\rho_{\max}=\max\left(\max_{(i,j)\in\{1,\ldots,p\}^2} (\bo S-\bo I_p)_{ij} -\min_{(i,j)\in\{1,\ldots,p\}^2} (\bo S-\bo I_p)_{ij}\right).
\end{align*}
The smallest value of the grid is then a tenth of the largest value $\rho_{\min}=0.1\rho_{\max}$. To obtain a logarithmic spaced grid, ten equally spaced values between $\log(\rho_{\min})$ and $\log(\rho_{\max})$ are transformed via the exponential function. We will use this grid of $\rho$-values in the remainder of the paper.

In general, the BIC criterion can be computed faster than cross validation. However, BIC tends to select too sparse models in practice. In our opinion, the gain in accuracy when using cross validation is worth the increased computation time. Therefore, we will use five-fold cross validation in the remainder of the paper.

\section{Breakdown point}
\label{Oellerer:sec:BP}

In Section~\ref{Oellerer:sec:def}, we obtain precision matrix estimators by replacing the sample covariance matrix in the GLASSO with robust covariance matrices. It is not immediately clear if the cellwise robustness of the initial covariance estimator translates to cellwise robustness of the final precision matrix estimator. Theorem~\ref{Oellerer:theorem:GLASSO_BP} shows that the final precision matrix estimator $\hat{\bs\Theta}_{\bo S}$ indeed inherits the breakdown point of the covariance matrix estimator $\bo S$. Furthermore, we formally show in Proposition~\ref{Oellerer:proposition:BPini} that the proposed initial covariance matrix estimators based on pairwise correlations are cellwise robust.

One of the most common measurements of robustness is the finite-sample breakdown point. We refer to \cite{Maronna} for the standard definition, i.e. under rowwise contamination. The breakdown point denotes the smallest amount of contamination in the data that drives the estimate to the boundary of the parameter space. For example, a location estimator needs to stay bounded, a dispersion estimator needs to stay bounded and away from zero. More formally, define for any symmetric $p\times p$ matrices $\bo A$ and $\bo B$
\begin{align*}
D(\bo A,\bo B)=\max\{|\lambda_1(\bo A)-\lambda_1(\bo B)|, |\lambda_p(\bo A)^{-1}-\lambda_p(\bo B)^{-1}|\},
\end{align*}
where the ordered eigenvalues of a matrix $\bo A$ are denoted by $0\leq\lambda_p(\bo A)\leq\ldots\leq\lambda_1(\bo A)$. We define the \textit{finite-sample breakdown point under cellwise contamination} of a precision matrix estimate $\hat{\bs \Theta}$ as 
\begin{align}
\label{Oellerer:eq:BP}
\epsilon_n(\hat{\bs \Theta},\bo X)=\min_{m=1,\ldots, n}\{\frac{m}{n}:\sup_{\bo X^m} D(\hat{\bs \Theta}(\bo X), \hat{\bs \Theta}(\bo X^m))=\infty\},
\end{align}
where $\bo X^m$ denotes a corrupted sample obtained from $\bo X\in\mathbb{R}^{n\times p}$ by replacing in each column at most $m$ cells by arbitrary values. Similarly, we can define the \textit{explosion} finite-sample breakdown point under cellwise contamination of a covariance matrix estimate $\bo S$ as 
\begin{align}
\label{Oellerer:eq:expBP}
\epsilon^+_n(\bo S,\bo X)=\min_{m=1,\ldots, n}\{\frac{m}{n}:\sup_{\bo X^m} |\lambda_1(\bo S(\bo X))-\lambda_1(\bo S(\bo X^m))|=\infty\},
\end{align}
where $\bo X^m$ denotes a corrupted sample obtained from $\bo X$ by replacing in each column at most $m$ cells by arbitrary values.

Finally, recall the definition of the explosion breakdown point of a univariate scale estimator $\scale(\cdot)$:
\begin{align*}
\epsilon^+_n(\scale,\mathbf{x})=\min_{m=1,\ldots, n}\{\frac{m}{n}:\sup_{\mathbf{x}^m}\scale(\mathbf{x}^m)=\infty\},
\end{align*}
where $\mathbf{x}^m$ is obtained from $\mathbf{x}\in\mathbb{R}^n$ by replacing $m$ of the $n$ values by arbitrary values.

To proof the main theorem of this section, we use different properties of eigenvalues, which we summarize in the following lemma.
\begin{lemma}
Let $\bo A,\bo B\in\mathbb{R}^{p\times p}$ and denote their smallest (largest) eigenvalues by $\lambda_p(\bo A)$ ($\lambda_1(\bo A)$) and $\lambda_p(\bo B)$ ($\lambda_1(\bo B)$), respectively. Then the following statements are true:
\begin{enumerate}[label=(\alph*)]
\item If $\bo A$ and $\bo B$ are positive semidefinite, then \label{Oellerer:prod}
\begin{align}
\label{Oellerer:eq:prod}
\lambda_p(\bo A\bo B)\leq \lambda_1(\bo A)\lambda_p(\bo B),\\
\label{Oellerer:eq:prod2}
\lambda_p(\bo A)\lambda_p(\bo B)\leq \lambda_p(\bo A\bo B).
\end{align}
\item If $\bo A$ and $\bo B$ are symmetric, then \label{Oellerer:sum}
\begin{align}
\label{Oellerer:eq:sum}
\lambda_1(\bo A+\bo B)=\lambda_1(\bo A) + \lambda_1(\bo B).
\end{align}
\item Denoting $\bo A=(a_{ij})_{i,j=1,\ldots,p}$, we have \label{Oellerer:max}
\begin{align}
\label{Oellerer:eq:max}
|\lambda_1(\bo A)|\leq p\max_{i,j=1,\ldots,p}|a_{ij}|.
\end{align}
\end{enumerate}
\end{lemma}

\begin{proof}
\ref{Oellerer:prod} \cite{Seber} 6.76 \ref{Oellerer:sum} \cite{Seber} 6.71, \ref{Oellerer:max} \cite{Seber} 6.26a
\end{proof}

Now, we can show that replacing the sample covariance matrix in the GLASSO by a robust covariance matrix $\bo S$ leads to a precision matrix estimator $\hat{\bs\Theta}_{\bo S}(\bo X)$ that inherits its robustness from $\bo S$.

\begin{theorem}
\label{Oellerer:theorem:GLASSO_BP} 
The finite sample breakdown point under cellwise contamination of the robust precision matrix estimator $\hat{\bs\Theta}_{\bo S}(\bo X)$ fulfills 
\begin{align}
\label{Oellerer:eq:GLASSO_BP}
\epsilon_n(\hat{\bs\Theta}_{\bo S},\bo X)\geq\epsilon^+_n(\bo S,\bo X)
\end{align}
with $\bo S$ a positive semidefinite covariance estimator.
\end{theorem}

\begin{proof}
Let $1 \leq  m \leq n$ be the maximum number of cells in a column that have been replaced to arbitrary positions. Since $\bo S(\bo X^m)$ is positive semidefinite, $\hat{\bs \Theta}_{\bo S}(\bo X^m)$ is positive definite \citep[see][Theorem 3]{Banerjee}. The estimate $\hat{\bs \Theta}_{\bo S}(\bo X^m)$ needs to fulfill the first order condition of (\ref{Oellerer:eq:GLASSOrob}):
\begin{align}
\label{Oellerer:eq:glasso_FOC}
\mathbf{0}=\hat{\bs \Theta}^{-1}_{\bo S}(\bo X^m)-\bo S(\bo X^m)-\rho\Sign\hat{\bs \Theta}_{\bo S}(\bo X^m),
\end{align}
where $(\Sign \hat{\bs \Theta}_{\bo S}(\bo X^m))_{jk}= \sign \hat{\bs \Theta}_{\bo S}(\bo X^m)_{jk}$ for $j,k=1,\ldots,p$. If $\hat{\bs\Theta}_{\bo S}$ has zero components, the first order condition (\ref{Oellerer:eq:glasso_FOC}) corresponds to a subdifferential and the sign function at 0 needs to be interpreted as the set $[-1,1]$ \citep{Bertsekas}. We then obtain
\begin{align*}
\bo I_p=(\bo S(\bo X^m)+\rho\Sign\hat{\bs \Theta}_{\bo S}(\bo X^m))\hat{\bs \Theta}_{\bo S}(\bo X^m).
\end{align*}
Thus, the smallest eigenvalue fulfills
\begin{align*}
1=\lambda_p(\bo I_p)=\lambda_p((\bo S(\bo X^m)+\rho\Sign\hat{\bs \Theta}_{\bo S}(\bo X^m))\hat{\bs \Theta}_{\bo S}(\bo X^m)).
\end{align*}
Using (\ref{Oellerer:eq:prod}), we get
\begin{align*}
1\leq \lambda_1(\bo S(\bo X^m)+\rho\Sign\hat{\bs \Theta}_{\bo S}(\bo X^m))\lambda_p(\hat{\bs \Theta}_{\bo S}(\bo X^m)).
\end{align*}
By definition $\hat{\bs \Theta}_{\bo S}(\bo X^m)$ is always symmetric, therefore also $\rho\Sign(\hat{\bs \Theta}_{\bo S}(\bo X^m))$. As a result, (\ref{Oellerer:eq:sum}) yields
\begin{align*}
1\leq[\lambda_1(\bo S(\bo X^m))+\lambda_1(\rho\Sign\hat{\bs \Theta}_{\bo S}(\bo X^m))]\lambda_p(\hat{\bs \Theta}_{\bo S}(\bo X^m)).
\end{align*}
As $\hat{\bs \Theta}_{\bo S}(\bo X^m)$ is positive definite, we obtain
\begin{align*}
\frac{1}{\lambda_p(\hat{\bs \Theta}_{\bo S}(\bo X^m))}\leq \lambda_1(\bo S(\bo X^m))+\rho\lambda_1(\Sign\hat{\bs \Theta}_{\bo S}(\bo X^m)).
\end{align*}
From the definition of the $\Sign$-function, we know that $|(\Sign\hat{\bs \Theta}_{\bo S}(\bo X^m))_{ij}|\leq1$. Together with (\ref{Oellerer:eq:max}), this yields
\begin{align}
\label{Oellerer:eq:aux1}
|\lambda_1(\Sign\hat{\bs \Theta}_{\bo S}(\bo X^m))|\leq p,
\end{align}
resulting in 
\begin{align}
\label{Oellerer:eq:ineq1}
\lambda_p(\hat{\bs \Theta}_{\bo S}(\bo X^m))^{-1}\leq\lambda_1(\bo S(\bo X^m))+\rho p.
\end{align}

From the definition of the explosion breakdown point (\ref{Oellerer:eq:expBP}), we know that for every $\tilde{m} < n \epsilon_n^+(\bo S,\bo X)$ there exists an $M <\infty$ such that
\begin{align}
\lambda_1(\bo S(\bo X^{\tilde{m}}))\leq M+\lambda_1(\bo S(\bo X))\label{Oellerer:eq:bps}.
\end{align}
Using (\ref{Oellerer:eq:ineq1}) in (\ref{Oellerer:eq:bps}) yields
\begin{gather*}
0\leq\lambda_p(\hat{\bs \Theta}_{\bo S}(\bo X^{\tilde{m}}))^{-1}\leq \lambda_1(\bo S(\bo X^{\tilde{m}}))+\rho p\leq   M+\lambda_1(\bo S(\bo X))+\rho p.
\end{gather*}
Together with the triangle inequality this gives
\begin{align}
\label{Oellerer:eq:aux_lp}
|\lambda_p(\hat{\bs \Theta}_{\bo S}(\bo X^{\tilde{m}}))^{-1}-\lambda_p(\hat{\bs \Theta}_{\bo S}(\bo X))^{-1}|&\leq \lambda_p(\hat{\bs \Theta}_{\bo S}(\bo X^{\tilde{m}}))^{-1}+\lambda_p(\hat{\bs \Theta}_{\bo S}(\bo X))^{-1}\\
&\leq M+\lambda_1(\bo S(\bo X))+\lambda_p(\hat{\bs \Theta}_{\bo S}(\bo X))^{-1}+\rho p.
\end{align}

To obtain a bound for the largest eigenvalue $\lambda_1(\hat{\bs\Theta}_S(\bo X^m))$, denote for any matrix $\bs \Theta\succ0$ 
\begin{align*}
Q(\bs \Theta, \bo X)=\log\det\bs\Theta-\tr(\bo S(\bo X)\bs\Theta)-\rho\sum_{j,k=1}^p|\theta_{jk}|.
\end{align*}
For the identity matrix, we obtain for contaminated data $\bo X^m$
\begin{align*}
Q(\bo I_p,\bo X^m)=0-\tr(\bo S(\bo X^m))-\rho p\geq-p\lambda_1(\bo S(\bo X^m))-\rho p
\end{align*}
since $\tr(\bo A)=\sum_{j=1}^p\lambda_j(\bo A)\leq p\lambda_1(\bo A)$ for any  matrix $\bo A\in\mathbb{R}^{p\times p}$. Using Equation (\ref{Oellerer:eq:bps}), this leads to 
\begin{align*}
Q(\bo I_p,\bo X^{\tilde{m}})\geq-pM-p\lambda_1(\bo S(\bo X))-\rho p.
\end{align*}

For any matrix $\tilde{\bs \Theta}\succ 0$, we obtain with (\ref{Oellerer:eq:prod2})
\begin{align}
\label{Oellerer:eq:myaux1}
\tr(\bo S\tilde{\bs \Theta})=\sum_{j=1}^p\lambda_j(\bo S\tilde{\bs \Theta})\geq\lambda_p(\bo S\tilde{\bs \Theta})\geq\lambda_p(\bo S)\lambda_p(\tilde{\bs \Theta})\geq 0.
\end{align}
Furthermore, (\ref{Oellerer:eq:max})  yields
\begin{align}
\label{Oellerer:eq:myaux2}
\sum_{i,j=1}^p|\tilde{\theta}_{jk}|\geq \max_{j,k=1,\ldots,p}|\tilde{\theta}_{jk}|\geq \frac{1}{p}\lambda_1(\tilde{\bs \Theta}).
\end{align}
Equations (\ref{Oellerer:eq:myaux1}) and (\ref{Oellerer:eq:myaux2}) lead to
\begin{align*}
Q(\tilde{\bs \Theta},\bo X^m)&=\log\det\tilde{\bs \Theta}-\tr(S\tilde{\bs \Theta})-\rho\sum_{j,k=1}^p|\theta_{jk}|\\
&\leq p \log \lambda_1(\tilde{\bs \Theta})-\frac{\rho}{p}\lambda_1(\tilde{\bs \Theta})
\end{align*}
because $\det(\bo A)=\prod_{j=1}^p\lambda_j(\bo A)\leq \lambda_1(\bo A)^p$ for any matrix $\bo A\in\mathbb{R}^{p\times p}$.

The function $x \mapsto p\log x-\rho x/p$ is concave and attains its maximum at $x=p^2/\rho$. Therefore, there exists a finite constant $M^\ast>p^2/\rho$, such that 
\begin{align*}
p\log M^\ast-\frac{\rho}{p}M^\ast=-pM-p\lambda_1(\bo S(\bo X))-\rho p.
\end{align*}
As a results, we know that any matrix $\tilde{\bs \Theta}$ with $\lambda_1(\tilde{\bs \Theta})>M^\ast$ is not optimizing (\ref{Oellerer:eq:GLASSOrob}) since $Q(\bo I_p, \bo X^{\tilde{m}})>Q(\tilde{\bs \Theta}, \bo X^{\tilde{m}})$. Hence,
\begin{align*}
0\leq\lambda_1(\hat{\bs\Theta}_{\bo S}(\bo X^{\tilde{m}}))\leq M^\ast.
\end{align*}
Together with the triangular inequality, this yields
\begin{align}
\label{Oellerer:eq:aux_l1}
|\lambda_1(\hat{\bs\Theta}_{\bo S}(\bo X^{\tilde{m}}))-\lambda_1(\hat{\bs\Theta}_{\bo S}(\bo X))|&\leq\lambda_1(\hat{\bs\Theta}_{\bo S}(\bo X^{\tilde{m}}))+\lambda_1(\hat{\bs\Theta}_{\bo S}(\bo X))\\
&\leq M^\ast+\lambda_1(\hat{\bs\Theta}_{\bo S}(\bo X)).
\end{align}

Thus, (\ref{Oellerer:eq:aux_lp}) and (\ref{Oellerer:eq:aux_l1}) lead to
\begin{align*}
\sup_{\bo X^{\tilde{m}}}D(\hat{\bs \Theta}_{\bo S}(\bo X), \hat{\bs \Theta}_{\bo S}(\bo X^{\tilde{m}}))\leq \max\{M+\lambda_1(\bo S(\bo X))+\rho p + \lambda_p(\hat{\bs \Theta}_{\bo S}(\bo X))^{-1},  M^\ast+\lambda_1(\hat{\bs \Theta}_{\bo S}(\bo X))\}
\end{align*}
for any $\tilde{m} < n \epsilon_n^+(\bo S,\bo X)$, yielding (\ref{Oellerer:eq:GLASSO_BP}).
\end{proof}

We still need to verify that the covariance matrix estimator based on pairwise correlations has a high explosion breakdown point under cellwise contamination.

\begin{proposition}
\label{Oellerer:proposition:BPini}
The explosion breakdown point under cellwise contamination of the covariance estimator based on pairwise correlations as defined in (\ref{Oellerer:eq:pairwise}) depends on the explosion breakdown point of the scale estimator used
\begin{align}
\label{Oellerer:eq:BPini}
\epsilon^+_n(\bo S, \bo X)\geq\max_{j=1,\ldots,p}\epsilon_n^+(\scale, \mathbf{x}^j).
\end{align}
\end{proposition}

\begin{proof}
Using the triangular inequality, (\ref{Oellerer:eq:max}), (\ref{Oellerer:eq:pairwise}) and the fact that a correlation has an absolute value smaller than 1, we obtain
\begin{align*}
|\lambda_1(\bo S(\bo X))-\lambda_1(\bo S(\bo X^m))|\leq |\lambda_1(\bo S(\bo X))|+p\max_{j,k=1,\ldots,p}|\scale((\bo X^m)^j)||\scale((\bo X^m)^k)|
\end{align*}
for any $m\in\{1,\ldots,n\}$, where $(\bo X^m)^j$ denots the $j$th column of matrix $\bo X^m$, and therefore (\ref{Oellerer:eq:BPini}).
\end{proof}

Note that the explosion breakdown point of the scale estimator in (\ref{Oellerer:eq:BPini}) is the breakdown point of a univariate estimator. Breakdown points of scale estimators have been studied extensively \citep[see e.g.][]{RousseeuwQn}. The median absolute deviation as well as the $Q_n$-estimator have an explosion breakdown point of $50\%$, resulting in a breakdown point of $50\%$ under cellwise contamination for the correlation based precision matrix estimator proposed in Section~\ref{Oellerer:sec:def}.

\section{Simulations}
\label{Oellerer:sec:sim}

In this section, we present a simulation study to compare the performance of the estimators introduced in Section~\ref{Oellerer:sec:def}. For the correlation based precision matrix estimator, we choose the $Q_n$-estimator as a scale. As robust correlation, we use Gaussian rank correlation, Spearman correlation and Quadrant correlation, resulting in the three different estimators `GlassoGaussQn', `GlassoSpearmanQn' and `GlassoQuadQn', respectively. As a point of reference, we also include the nonrobust, classical GLASSO (\ref{Oellerer:eq:GLASSO}) and abbreviate it as `GlassoClass'. Additionally, we compute a covariance based precision matrix estimate, where we choose $Q_n$ as the scale estimator and NPD to obtain a positive semidefinite covariance estimate (`GlassoNPDQn'). This estimator represents the class of estimators studied by \cite{Tarr2}. 

To compare to a rowwise, but not cellwise robust estimator that can be computed in high dimensions, we consider the spatial sign covariance matrix \citep{Oja}
\begin{align}
\label{Oellerer:eq:spatialini}
\bo S_{sign}^{incons}(\bo X)=\frac{1}{n}\sum_{i=1}^n U(\mathbf{x}_i-\mathbf{\hat{\mu}})U(\mathbf{x}_i-\mathbf{\hat{\mu}})^\top,
\end{align}
where $U(\mathbf{y})=\|\mathbf{y}\|_2^{-1}\mathbf{y}$ if $\mathbf{y}\neq \bo 0$ and $U(\mathbf{y})=0$ otherwise, and $\|\mathbf{y}\|_2$ stands for the Euclidean norm. The location estimator $\mathbf{\hat{\mu}}$ is the spatial median, i.e. the minimizer of $\sum_{i=1}^n\|\mathbf{x}_i-\mathbf{\mu}\|_2$. Since only the eigenvectors of (\ref{Oellerer:eq:spatialini}) are consistent estimators for the eigenvectors of the covariance matrix at the normal model, we still need to compute consistent eigenvalues. Let $\bo U$ denote the matrix of eigenvectors of (\ref{Oellerer:eq:spatialini}). The eigenvalues of the covariance matrix are then given by the marginal variances of $\bo U^\top\mathbf{x}_1,\ldots,\bo U^\top\mathbf{x}_n$. To robustly estimate these marginal variances, we use the robust scale estimator $Q_n$. Denote the matrix of robust eigenvalues as $\bs \Lambda=\diag(\hat{\lambda}_1,\ldots,\hat{\lambda}_p)$. Then the consistent spatial sign covariance matrix is 
\begin{align*}
\bo S_{sign}(\bo X)=\bo U \bs \Lambda \bo U^\top.
\end{align*}
The spatial sign covariance matrix is positive semidefinite. Therefore, we use it as an input in the GLASSO, as in Equation (\ref{Oellerer:eq:GLASSOrob}), to obtain a sparse precision matrix estimate which is robust under rowwise contamination. We refer to this precision matrix estimator as `GlassoSpSign'. Finally, we also add the inverse of the classical sample covariance matrix (\ref{Oellerer:eq:sample}) as a benchmark (`Classic'), where it can be computed. For all estimators, we select the regularization parameter $\rho$ via five-fold cross validation over a logarithmic spaced grid (see Section~\ref{Oellerer:sec:rho}).

\medskip
\textit{Sampling schemes.\quad}
We use in total four sampling schemes covering the scenarios of a banded precision matrix, a sparse precision matrix, a dense precision matrix \citep{Cai2} and a diagonal precision matrix. Each sampling scheme is defined through the true precision matrix $\bs \Theta_0\in\mathbb{R}^{p\times p}$ for $i,j=1,\ldots,p$:
\begin{itemize}
\item `banded': $(\bs \Theta_0)_{ij}=0.6^{|i-j|}$
\item `sparse': $\bs \Theta_0=\bo B+\delta \bo I_p$ with $\mathbb{P}[b_{ij}=0.5]=0.1$ and $\mathbb{P}[b_{ij}=0]=0.9$ for $i\neq j$. The parameter $\delta$ is chosen such that the conditional number of $\bs \Theta_0$ equals $p$. Then the matrix is standardized to have unit diagonals.
\item `dense': $(\bs \Theta_0)_{ii}=1$ and $(\bs \Theta_0)_{ij}=0.5$ for $i\neq j$
\item `diagonal': $(\bs \Theta_0)_{ii}=1$ and $(\bs \Theta_0)_{ij}=0$ for $i\neq j$
\end{itemize}
For each sampling scheme, we generate $M=100$ samples of size $n=100$ from a multivariate normal $\mathcal{N}(0,\bs \Theta^{-1}_0)$. We take as dimension $p=60$ and $p=200$. 

\medskip
\textit{Contamination.\quad}
To simulate contamination, we use two different contamination settings \citep{Finegold}: (i) To every generated data set, we add $5$ or $10\%$ of cellwise contamination. Therefore, we randomly select $5$ and $10\%$ of the cells and draw them from a normal $\mathcal{N}(10, 0.2)$. (ii) To simulate model deviation, we draw all observations from an alternative $t$-distribution $t^\ast_{100,2}(\mathbf{0},\bs \Theta_0^{-1})$ of dimension $100$ with 2 degree of freedom. 

Recall that a multivariate $t$-distributed random variable $\mathbf{x}\sim t_{n,\nu}(\mathbf{0}, \bs \Psi)$ is defined as a multivariate normally distributed random variable $\mathbf{y}=(y_1,\ldots,y_p)^\top\sim\mathcal{N}_p(\bo 0, \bs \Psi)$ divided by a gamma distributed variable $\tau\sim \Gamma(\nu/2, \nu/2)$
\begin{align*}
\mathbf{x}=\frac{\mathbf{y}}{\sqrt{\tau}}.
\end{align*}
To obtain an alternative $t$-distributed random variable $\mathbf{x}=(x_1,\ldots, x_p)^\top\sim t^\ast_{n, \nu}(\mathbf{0}, \bs\Psi)$, we draw $p$ independent divisors $\tau_j\sim \Gamma(\nu/2, \nu/2)$ for the different variables $j=1,\ldots,p$ 
\begin{align*}
x_j = \frac{y_j}{\sqrt{\tau_j}}.
\end{align*}
The heaviness of the tails is then different for different variables of $\mathbf{x}$.

\medskip
\textit{Performance measures.\quad}
We assess the performance of the estimators using the Kullback-Leibler divergence \citep[][p437]{Buhlmann}
\begin{align*}
KL(\hat{\bs \Theta}, \bs \Theta_0)=\tr(\bs \Theta_0^{-1}\hat{\bs \Theta})-\log \det(\bs \Theta_0^{-1}\hat{\bs \Theta})-p.
\end{align*}
It measures how close the obtained estimate $\hat{\bs\Theta}$ is to the true parameter $\bs \Theta_0$. Lower values represent a better estimate. If the estimator is equal to the true precision matrix, the Kullback-Leibler distance is equal to zero. The less accurate the precision matrix is estimated, the higher the value of the Kullback-Leibler distance becomes.

To measure how well the sparseness of the true precision matrix is recovered, we also look at false positive (FP) and false negative (FN) rates:
\begin{align*}
FP&=\frac{|\{(i,j): i=1,\ldots,n; j=1,\ldots,p:(\hat{\bs \Theta})_{ij} \neq 0\land (\bs \Theta_0)_{ij} = 0\}|}{|\{(i,j): i=1,\ldots,n; j=1,\ldots,p:(\bs \Theta_0)_{ij} = 0\}|}\\
FN&=\frac{|\{(i,j): i=1,\ldots,n; j=1,\ldots,p:(\hat{\bs \Theta})_{ij} = 0\land (\bs \Theta_0)_{ij} \neq 0\}|}{|\{(i,j): i=1,\ldots,n; j=1,\ldots,p:(\bs \Theta_0)_{ij} \neq 0\}|}
\end{align*}
The false positive rate gives the percentage of zero-elements in the true precision matrix that are wrongly estimated as nonzero. In contrast, the false negative rate gives the percentage of nonzero-elements in the true precision matrix that are wrongly estimated to be zero. Both values are desired to be as small as possible. However, a large false negative rate has a worse impact since it implies that associations between variables are not found and therefore important information is not used. A large false positive rate indicates that unnecessary associations are included, which `only' complicates the model. Note that if $\bs \Theta_0$ does not contain any zero-entries, the false positive rate is not defined. In graphical modeling, a high false negative rate indicates that many non-zero edges that should be included in the estimated graph are missed. This implies that there are conditional independencies assumed which are not supported by the true graph.

\medskip
\textit{Simulation results.\quad}
Results for $p=60$ are given in Table~\ref{Oellerer:tab:low}. For clean data in the banded scenario, the classical GLASSO (`GlassoClass') is performing best, achieving lowest values of KL. Only marginally higher values of KL are obtained by the correlation based precision matrix using Gaussian rank correlation (`GlassoGaussQn') and the regularized spatial sign covariance matrix (`GlassoSpSign'). Their good performance can be explained by their high efficiency at the normal model. Even though this data is clean, the inverse of the sample covariance matrix (`Classic') is performing very poorly. This is due to the low precision of the sample covariance matrix for a data set with $p>n/2$. Regularization of the inverse of the sample covariance matrix is solving the problem, as we see from the classical GLASSO. Note that the sample covariance matrix always gives an FN of zero, since the resulting estimate is not sparse, and should therefore not be considered to evaluate the performance of the sample covariance matrix. The correlation based precision matrix using Spearman correlation (`GlassoSpearmanQn') obtains a slightly higher value of KL than `GlassoGaussQn'. It probably suffers from its inconsistency. This also explains why the KL of the correlation based precision matrix using Quadrant correlation (`GlassoQuadQn') is so much higher, since the asymptotic bias of the Quadrant correlation is considerably higher than that of Spearman. The performance of the covariance based precision matrix (`GlassoNPDQn') lies in between `GlassoSpearmanQn' and `GlassoQuadQn'. 

Under contamination, the relative performance of the different estimators changes. Clearly, the classical GLASSO is not robust, and it achieves the highest values of KL of all estimators. Also the regularized spatial sign covariance matrix does not perform well. This is no surprise since for 5\% of cellwise contamination, already more than 90\% of the observations are expected to be contaminated. Thus, the level of rowwise contamination is too high for `GlassoSpSign' to obtain reliable results. Best performance under contamination is obtained by the correlation based precision matrices using Gaussian rank or Spearman correlation. They give lowest values of KL for all three contamination schemes. Moderately larger values are obtained by `GlassoQuadQn'. Of the cellwise robust estimators, the covariance based precision matrix estimator is performing worst under contamination. It obtains highest values of KL and FN in all three contamination settings. Under 10\% of cellwise contamination the value of KL of `GlassoNPDQn' is nearly double that of `GlassoSpearmanQn'. 

Looking at the other three sampling schemes `sparse', `dense' and `diagonal', the conclusions are very similar to that of the banded scheme: For clean data `GlassoClass' is doing best, closely followed by `GlassoGaussQn' and `GlassoSpSign'. Under contamination `GlassoGaussQn' and `GlassoSpearmanQn' are performing best, while `GlassoNPDQn' gives worst results of all cellwise robust estimators. For the sparse settings `sparse' and `diagonal' we also compare the different values of the FP and FN. In the setting `diagonal' the values are more or less the same for all estimators (apart from the sample covariance matrix which does not give sparse results and therefore has a FP equal to one). In the setting `sparse', differences are more outspoken. The covariance based precision matrix estimator gives a FN of up to double that of `GlassoGaussQn' or `GlassoSpearmanQn', which is not made up by the slightly lower value of FP. In graphical modeling that means that many nonzero edges are missed by `GlassoNPDQn', while they are correctly identified by `GlassoGaussQn' and `GlassoSpearmanQn'. 
\smallskip

\begin{table}[t]
\caption{Simulation results for $n=100$ and $p=60$: Kullback-Leibler criterion (KL),  false positive rate (FP) and false negative rate (FN) averaged over $M=100$ simulations reported for 7 estimators and 4 sampling schemes}
\label{Oellerer:tab:low}
{
\resizebox{0.75\textwidth}{!}{\begin{minipage}{\textwidth}
\centering
\begin{tabular}{|l|l||c|c|c||c|c|c||c|c|c||c|c|c|}
\hline
&& \multicolumn{3}{c||}{clean} & \multicolumn{3}{c||}{ 5\% cellwise} & \multicolumn{3}{c||}{ 10\% cellwise}  & \multicolumn{3}{c|}{ alternative $t$} \\ 
\hline
& & KL & FP & FN & KL & FP & FN & KL & FP & FN & KL & FP & FN \\ 
\hline
\multirow{7}{*}{banded} & GlassoClass & 8.97 &  & .70 & 55.00 &  & .95 & 77.11 &  & .94 & 143.16 &  & .98 \\ 
& GlassoQuadQn & 14.96 &  & .83 & 19.20 &  & .86 & 24.44 &  & .90 & 31.10 &  & .87 \\ 
& GlassoGaussQn & 9.62 &  & .75 & 16.91 &  & .83 & 23.52 &  & .88 & 28.41 &  & .84 \\ 
& GlassoSpearmanQn & 10.09 &  & .76 & 16.32 &  & .83 & 22.69 &  & .87 & 27.92 &  & .84 \\ 
& GlassoNPDQn & 11.90 &  & .85 & 21.73 &  & .91 & 43.59 &  & .97 & 37.34 &  & .92 \\ 
& Classic & 71.54 &  & .00 & 49.24 &  & .00 & 61.01 &  & .00 & 67.84 &  & .00 \\ 
& GlassoSpSign & 9.53 &  & .74 & 53.92 &  & .96 & 77.54 &  & .95 & 80.41 &  & .94 \\ 
\hline
\multirow{7}{*}{sparse} & GlassoClass & 5.87 & .23 & .09 & 63.70 & .02 & .82 & 88.81 & .04 & .81 & 140.09 & .00 & .85 \\ 
& GlassoQuadQn & 10.28 & .15 & .38 & 14.20 & .12 & .47 & 19.04 & .09 & .56 & 26.28 & .10 & .45 \\ 
& GlassoGaussQn & 6.34 & .21 & .11 & 12.25 & .16 & .30 & 18.39 & .11 & .49 & 24.09 & .13 & .28 \\ 
& GlassoSpearmanQn & 6.74 & .21 & .13 & 11.75 & .16 & .27 & 17.67 & .12 & .43 & 23.71 & .14 & .26 \\ 
& GlassoNPDQn & 8.25 & .13 & .23 & 17.85 & .06 & .47 & 42.14 & .01 & .82 & 32.73 & .06 & .52 \\ 
& Classic & 71.54 & 1.00 & .00 & 49.39 & 1.00 & .00 & 66.83 & 1.00 & .00 & 62.79 & 1.00 & .00 \\ 
& GlassoSpSign & 6.35 & .21 & .11 & 62.36 & .01 & .83 & 89.37 & .02 & .83 & 76.37 & .04 & .65 \\ 
\hline
\multirow{7}{*}{dense} & GlassoClass & 4.40 &  & .92 & 42.52 &  & .96 & 64.95 &  & .94 & 128.00 &  & .98 \\ 
& GlassoQuadQn & 4.65 &  & .94 & 7.66 &  & .95 & 11.66 &  & .96 & 20.72 &  & .97 \\ 
& GlassoGaussQn & 4.59 &  & .93 & 7.59 &  & .94 & 11.70 &  & .96 & 20.72 &  & .96 \\ 
& GlassoSpearmanQn & 4.61 &  & .94 & 7.59 &  & .94 & 11.80 &  & .96 & 20.81 &  & .97 \\ 
& GlassoNPDQn & 5.01 &  & .96 & 13.69 &  & .98 & 30.98 &  & .98 & 28.43 &  & .98 \\ 
& Classic & 71.54 &  & .00 & 39.88 &  & .00 & 49.44 &  & .00 & 59.08 &  & .00 \\ 
& GlassoSpSign & 4.62 &  & .94 & 41.65 &  & .97 & 65.21 &  & .96 & 69.46 &  & .98 \\ 
\hline
\multirow{7}{*}{diagonal} & GlassoClass & 1.31 & .05 & .00 & 66.11 & .01 & .00 & 93.48 & .03 & .00 & 124.03 & .00 & .00 \\ 
& GlassoQuadQn & 1.55 & .04 & .00 & 4.60 & .03 & .00 & 8.68 & .03 & .00 & 17.69 & .02 & .00 \\ 
& GlassoGaussQn & 1.53 & .04 & .00 & 4.54 & .04 & .00 & 8.67 & .03 & .00 & 17.61 & .02 & .00 \\ 
& GlassoSpearmanQn & 1.55 & .04 & .00 & 4.57 & .04 & .00 & 8.68 & .03 & .00 & 17.78 & .02 & .00 \\ 
& GlassoNPDQn & 1.92 & .02 & .00 & 11.02 & .00 & .00 & 33.94 & .00 & .00 & 25.46 & .00 & .00 \\ 
& Classic & 71.54 & 1.00 & .00 & 48.41 & 1.00 & .00 & 68.67 & 1.00 & .00 & 56.26 & 1.00 & .00 \\ 
& GlassoSpSign & 1.54 & .04 & .00 & 62.99 & .01 & .00 & 93.75 & .02 & .00 & 66.05 & .01 & .00 \\ 
\hline
\end{tabular}
\end{minipage}}
}
\end{table}

The simulation results for $p=200$ are given in Table~\ref{Oellerer:tab:high}. Since $p>n$, the sample covariance matrix cannot be inverted anymore and is excluded from the analysis. Overall, the conclusions are similar to $p=60$. For clean data, the classical GLASSO performs best. Marginally larger values of KL are obtained by `GlassoGaussQn' and `GlassoSpSign'. In comparision to $p=60$, here also `GlassoSpearmanQn' is doing very well for clean data. 

For $p=200$, we see again that under any type of contamination the classical GLASSO and the regularized spatial sign covariance matrix are not reliable any more. In contrast, the cellwise robust correlation based precision matrix estimators achieve very good results, especially in combination with Gaussian rank or Spearman correlation. Their KL as well as their FN are lowest of all estimators for all settings considered here. The covariance based correlation estimate is considerably less accurate than the correlation based estimates. Under higher amounts of cellwise contamination `GlassoNPDQn' can have a KL of more than four times the value of the correlation based precision matrix estimators. Besides, its FN is higher in all settings considered.

\begin{table}[t]
\caption{Simulation results for $n=100$ and $p=200$: Kullback-Leibler criterion (KL),  false positive rate (FP) and false negative rate (FN) averaged over $M=100$ simulations reported for 7 estimators and 4 sampling schemes} 
\label{Oellerer:tab:high}
{\resizebox{0.75\textwidth}{!}{\begin{minipage}{\textwidth}
\centering
  \begin{tabular}{|l|l||c|c|c||c|c|c||c|c|c||c|c|c|}
  \hline
  && \multicolumn{3}{c||}{clean} & \multicolumn{3}{c||}{ 5\% cellwise} & \multicolumn{3}{c||}{ 10\% cellwise}  & \multicolumn{3}{c|}{ alternative $t$} \\ 
  \hline
  & & KL & FP & FN & KL & FP & FN & KL & FP & FN & KL & FP & FN \\ 
  \hline
  \multirow{7}{*}{banded} & GlassoClass & 38.32 &  & .89 & 187.21 &  & .98 & 262.04 &  & .98 & Inf &  & .99 \\ 
  & GlassoQuadQn & 56.42 &  & .94 & 70.67 &  & .95 & 86.97 &  & .97 & 112.04 &  & .96 \\ 
  & GlassoGaussQn & 40.18 &  & .91 & 63.97 &  & .94 & 84.67 &  & .96 & 103.21 &  & .94 \\ 
  & GlassoSpearmanQn & 41.53 &  & .90 & 61.90 &  & .93 & 82.92 &  & .96 & 101.81 &  & .93 \\ 
  & GlassoNPDQn & 52.42 &  & .96 & 102.91 &  & .98 & 200.13 &  & .99 & 164.86 &  & .99 \\ 
  & Classic &  &  &  &  &  &  &  &  &  &  &  &  \\ 
  & GlassoSpSign & 39.68 &  & .90 & 189.47 &  & .98 & 265.32 &  & .98 & 326.08 &  & .99 \\ 
  \hline
  \multirow{7}{*}{sparse} & GlassoClass & 46.65 & .10 & .52 & 220.90 & .02 & .93 & 302.94 & .02 & .93 & Inf & .00 & .95 \\ 
  & GlassoQuadQn & 60.42 & .09 & .72 & 75.70 & .07 & .77 & 93.11 & .05 & .81 & 119.80 & .06 & .78 \\ 
  & GlassoGaussQn & 48.45 & .10 & .54 & 69.70 & .08 & .69 & 90.69 & .06 & .79 & 115.29 & .06 & .71 \\ 
  & GlassoSpearmanQn & 49.60 & .10 & .56 & 68.27 & .08 & .67 & 88.92 & .06 & .76 & 114.46 & .06 & .70 \\ 
  & GlassoNPDQn & 58.64 & .07 & .64 & 111.15 & .03 & .84 & 215.95 & .00 & .95 & 167.81 & .02 & .85 \\ 
  & Classic &  &  &  &  &  &  &  &  &  &  &  &  \\ 
  & GlassoSpSign & 47.97 & .10 & .54 & 223.28 & .01 & .93 & 306.49 & .01 & .94 & 339.53 & .01 & .93 \\ 
  \hline
  \multirow{7}{*}{dense} & GlassoClass & 9.70 &  & .97 & 137.73 &  & .98 & 214.08 &  & .98 & Inf &  & .99 \\ 
  & GlassoQuadQn & 10.41 &  & .98 & 21.12 &  & .98 & 35.06 &  & .98 & 66.07 &  & .99 \\ 
  & GlassoGaussQn & 10.35 &  & .98 & 20.91 &  & .98 & 35.06 &  & .98 & 65.54 &  & .99 \\ 
  & GlassoSpearmanQn & 10.39 &  & .98 & 21.11 &  & .98 & 34.92 &  & .98 & 65.80 &  & .99 \\ 
  & GlassoNPDQn & 15.27 &  & .99 & 65.43 &  & .99 & 146.61 &  & .99 & 121.14 &  & .99 \\ 
  & Classic &  &  &  &  &  &  &  &  &  &  &  &  \\ 
  & GlassoSpSign & 10.53 &  & .98 & 140.17 &  & .99 & 217.08 &  & .98 & 270.10 &  & .99 \\ 
  \hline
  \multirow{7}{*}{diagonal} & GlassoClass & 5.41 & .02 & .00 & 224.15 & .01 & .00 & 317.07 & .01 & .00 & Inf & .00 & .00 \\ 
  & GlassoQuadQn & 6.14 & .02 & .00 & 17.12 & .01 & .00 & 30.71 & .01 & .00 & 61.51 & .01 & .00 \\ 
  & GlassoGaussQn & 6.05 & .02 & .00 & 17.15 & .01 & .00 & 30.66 & .01 & .00 & 61.37 & .01 & .00 \\ 
  & GlassoSpearmanQn & 6.07 & .02 & .00 & 17.08 & .01 & .00 & 30.71 & .01 & .00 & 61.18 & .01 & .00 \\ 
  & GlassoNPDQn & 10.83 & .00 & .00 & 63.60 & .00 & .00 & 167.73 & .00 & .00 & 114.93 & .00 & .00 \\ 
  & Classic &  &  &  &  &  &  &  &  &  &  &  &  \\ 
  & GlassoSpSign & 6.29 & .01 & .00 & 225.91 & .01 & .00 & 320.17 & .01 & .00 & 265.02 & .00 & .00 \\ 
  \hline
  \end{tabular}
\end{minipage}}}
\end{table}

Since in high-dimensional analysis computation time is important for practical usage of the estimators, Table~\ref{Oellerer:tab:time} gives an overview of the average computation time that the different estimators require. The computation time was comparable throughout the different simulation schemes. Therefore, we only give averages. Note that the reported computation time includes the selection of $\rho$ via 5-fold crossvalidation. For $p=60$, the correlation based precision matrices, the classical GLASSO and the regularized spatial sign covariance matrix need very similar computation times. This indicates that the GLASSO algorithm takes most of the computation time and that the computation time of the initial covariance matrices is negligible. In contrast, the covariance based precision matrix estimator is nearly four times slower. For $p=200$, the classical GLASSO and the regularized spatial sign covariance matrix can be computed fastest. But as they are not robust enough, the estimates are very inaccurate. Computation of the correlation based estimators is still very fast here. The estimation including the selection of $\rho$ over a grid of ten values takes less than 10 seconds. In contrast, estimation of the covariance based precision matrix takes more than 20 times longer. 

\begin{table}[t]
\centering
\caption{Computation time (in sec.) for samples of size $n=100$ (including selection of $\rho$ via 5-fold cross-validation) averaged over $M=100$ simulations and all simulation schemes reported for 7 estimators}
\label{Oellerer:tab:time}
{\footnotesize
\begin{tabular}{|l|c|c|}
  \hline
 & $p=60$ & $p=200$ \\ 
  \hline
 GlassoClass & 5.93 & 7.69 \\ 
   GlassoQuadQn & 6.13 & 9.12 \\ 
   GlassoGaussQn & 6.09 & 9.15 \\ 
    GlassoSpearmanQn & 5.82 & 9.01 \\ 
    GlassoNPDQn & 22.85 & 216.79\\ 
    Classic & 0.00 & 0.02\\ 
    GlassoSpSign & 5.73 & 8.11\\ 
   \hline
\end{tabular}
}
\end{table}

Since we advertise the high breakdown point of the correlation based precision matrix estimators, we also look at the performance of the estimators under higher amounts of cellwise contamination, ranging from 0 to 40\%. Fig.~\ref{Oellerer:fig:1} plots the value of KL for the most representative precision matrix estimators for $p=60$ (left panel) and $p=200$ (right panel), following the `banded' sampling scheme. As expected, the nonrobust `GlassoClass' results in the highest values of KL. For higher amounts of cellwise contamination, the KL of the `GlassoNPDQn' deteriorates quickly. This is in contrast with the more robust `GlassoGaussQn', where the KL measure remains limited for higher contamination levels, both for $p=60$ and $p=200$. The results for the sampling schemes `sparse', `dense' and `diagonal' are comparable to Fig.~\ref{Oellerer:fig:1} and are therefore omitted.

\begin{figure}[t]
\centerline{
		\includegraphics[width=0.5\textwidth]{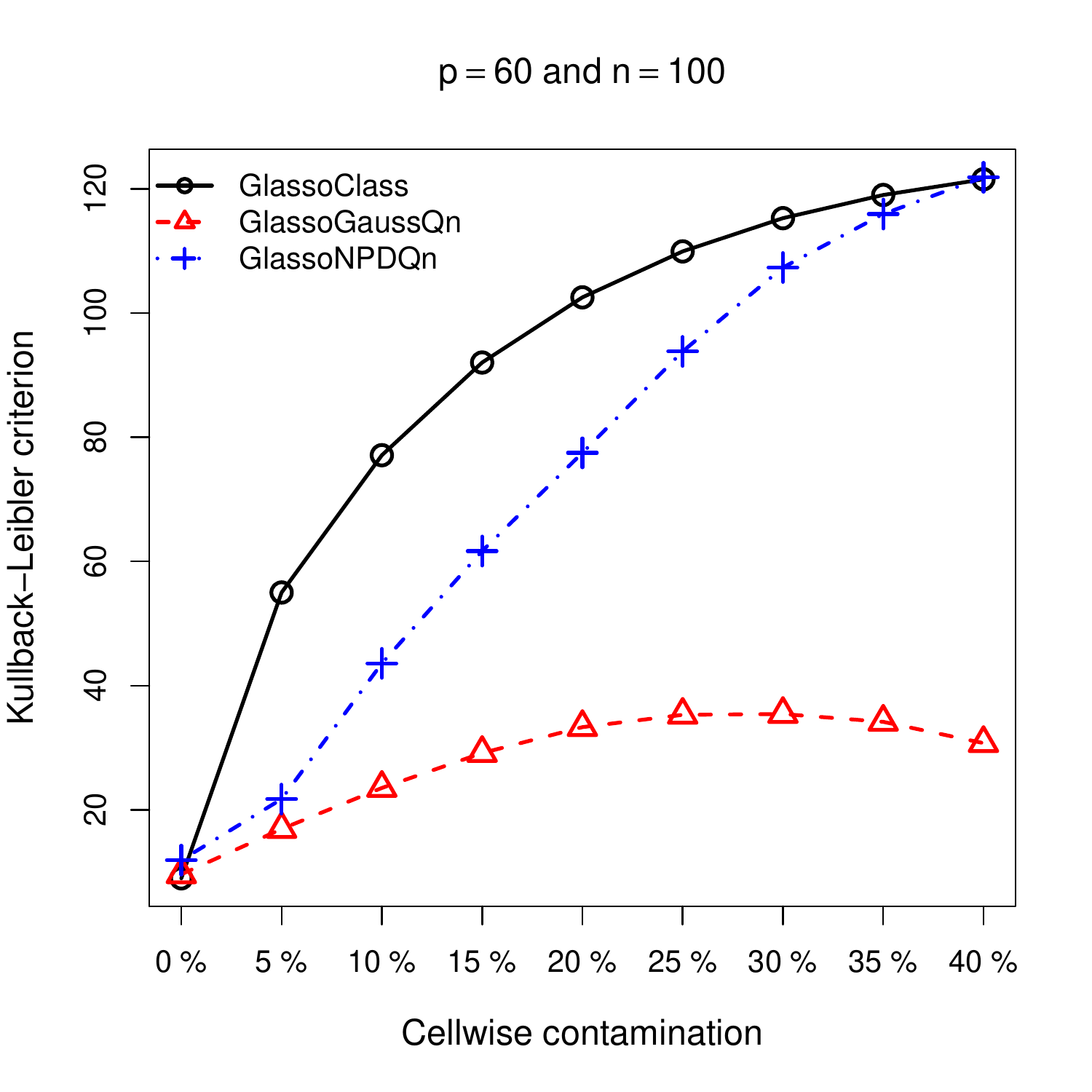}%
		\includegraphics[width=0.5\textwidth]{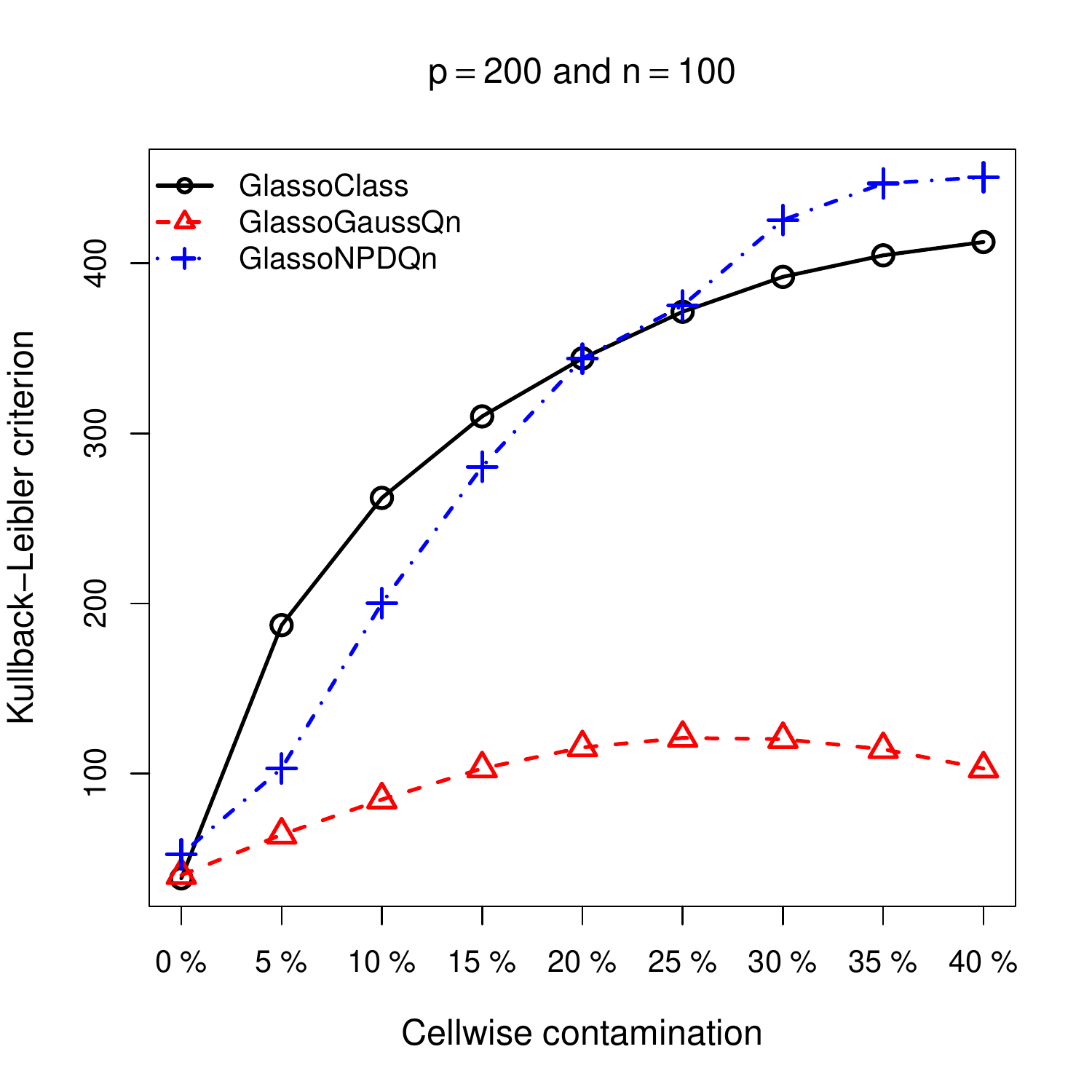}%
}
\caption{Kullback-Leibler criterion for the `banded' sampling scheme averaged over $M=100$ simulations reported for various amounts of cellwise contamination and several estimators}
\label{Oellerer:fig:1}
\end{figure}

To summarize, for clean data the classical GLASSO performs best. Under cellwise contamination, `GlassoGaussQn' and `GlassoSpearmanQn' achieve best results. All three estimators can be computed equally fast. Since the `GlassoGaussQn' is consistent and performs similarly well as the classical GLASSO for clean data, we advise the `GlassoGaussQn' for high-dimensional sparse precision matrix estimation under cellwise contamination.

\section{Applications}
\label{Oellerer:sec:app}

In this paper, we describe how a cellwise robust, sparse precision matrix estimator can be obtained. To show the applicability of the introduced estimator to a real world data set, we use the dataset \texttt{stockdata}, which is publicly available through the \textsf{R}-package \texttt{huge} \citep{huge}. It consists of the closing prices of $p=452$ stocks in the S\&P on all trading days between January 1,2003 and January 1, 2008, leading to $n=1258$ observations.
We use the same data transformations and parameter choices as in \cite{Zhao}. The estimated graphical models returned by `GlassoClass' and `GlassoGaussQn' are visualized in Panel (a) and (b) of Fig.~\ref{Oellerer:tab:network}. From the plots, we can conclude that the two graphs are very similar. Indeed, only around 2\% of the selected edges in `GlassoClass' are not selected in `GlassoGaussQn', while the percentage is even smaller vice versa. As a result, we assume that \texttt{stockdata} is a rather clean data set.

To see how the estimators behave under contamination, we randomly select 5\% of the cells of the data matrix and replace them by replicates of the normal distribution $\mathcal{N}(10, 0.2)$. The graphs estimated by `GlassoClass' and `GlassoGaussQn' from the contaminated data are shown in Panels (c) and (d) of Fig.~\ref{Oellerer:tab:network}, respectively. While the graph estimated by `GlassoGaussQn' hardly differs from the uncontaminated case, `GlassoClass' estimates a graph without any edges. Thus, `GlassoGaussQn' is robust in the sense that the estimate on the contaminated data resembles that of the clean data. In contrast, the nonrobust `GlassoClass' returns a not reliable estimate in the presence of cellwise contaminated data.

\begin{figure}
\resizebox{0.9\textwidth}{!}{\begin{minipage}{\textwidth}
\begin{tabular}{|c|c|c|c|}
\hline
\multicolumn{2}{|c|}{data} & \multicolumn{2}{c|}{contaminated data}\\
\hline
`GlassoClass' & `GlassoGaussQn' & `GlassoClass' & `GlassoGaussQn' \\
\hline
	\includegraphics[width=0.245\textwidth]{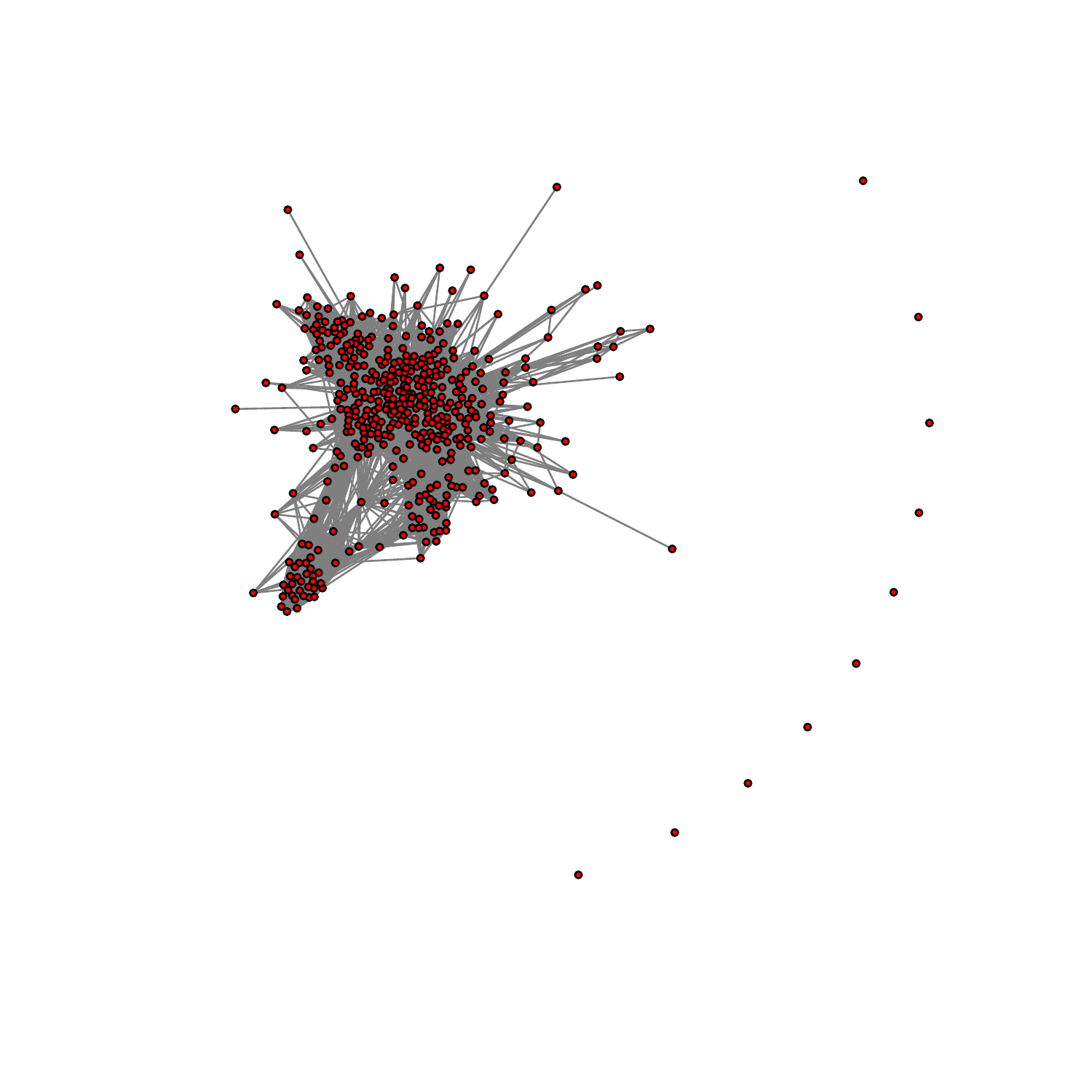} &
\includegraphics[width=0.245\textwidth]{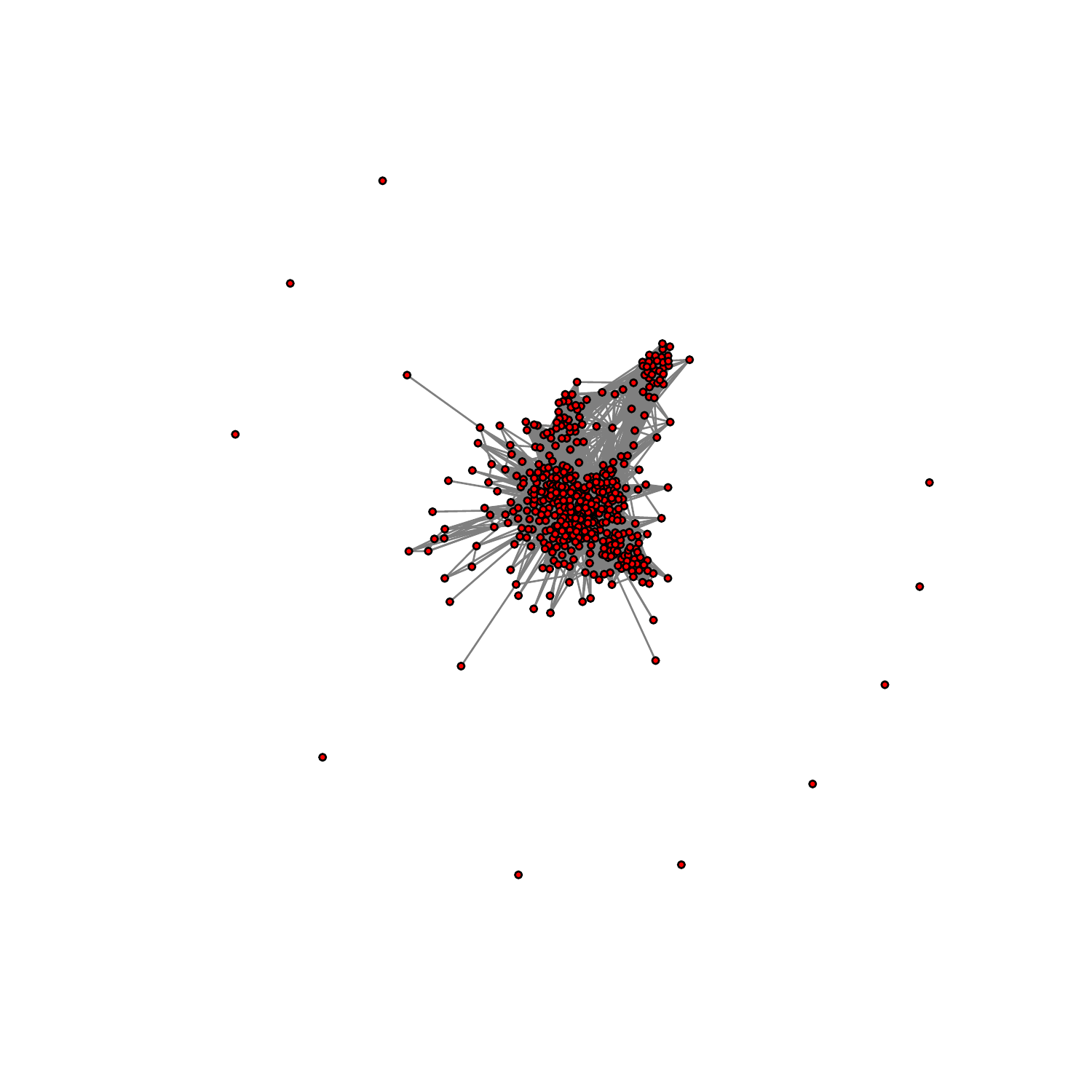} &
\includegraphics[width=0.245\textwidth]{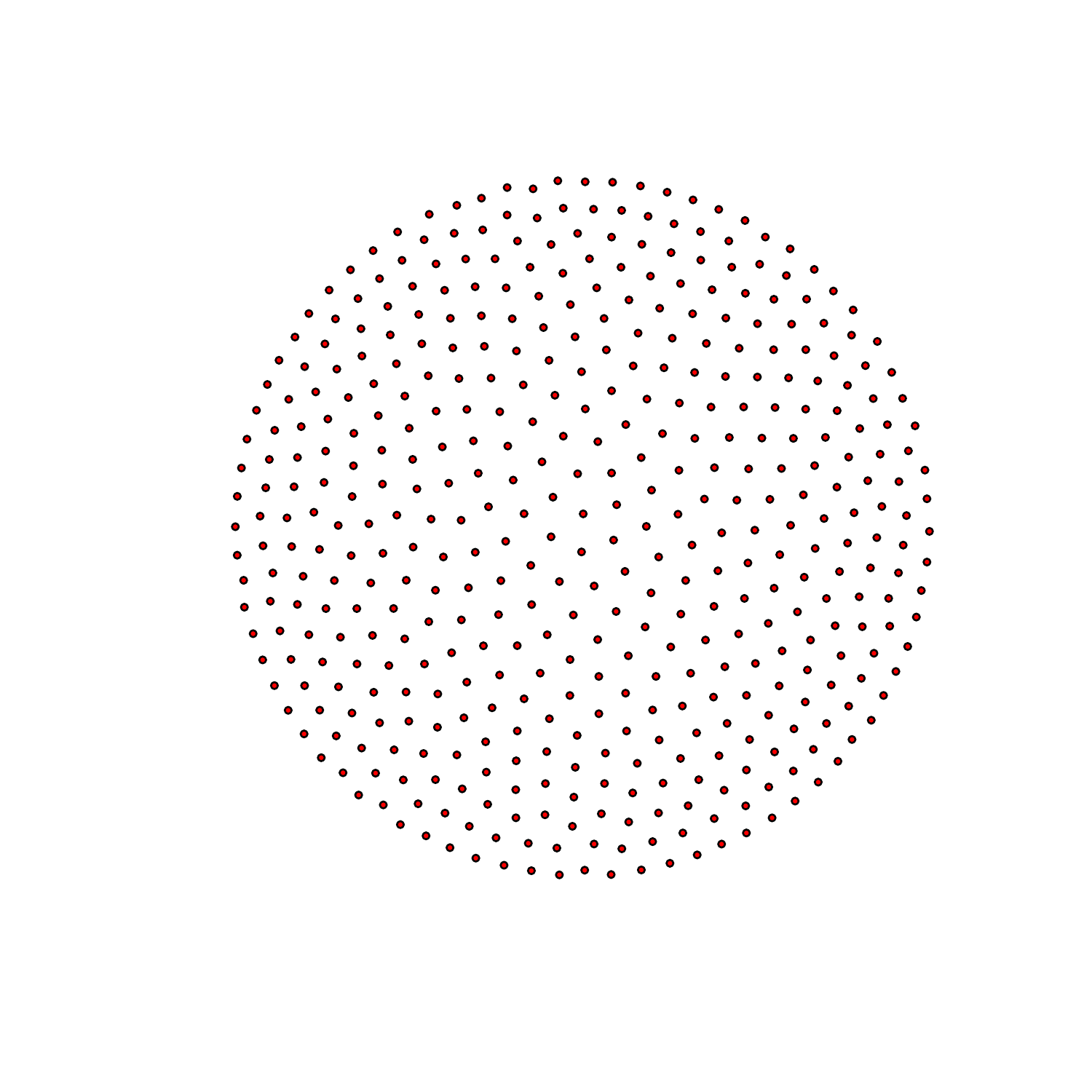} & 
\includegraphics[width=0.245\textwidth]{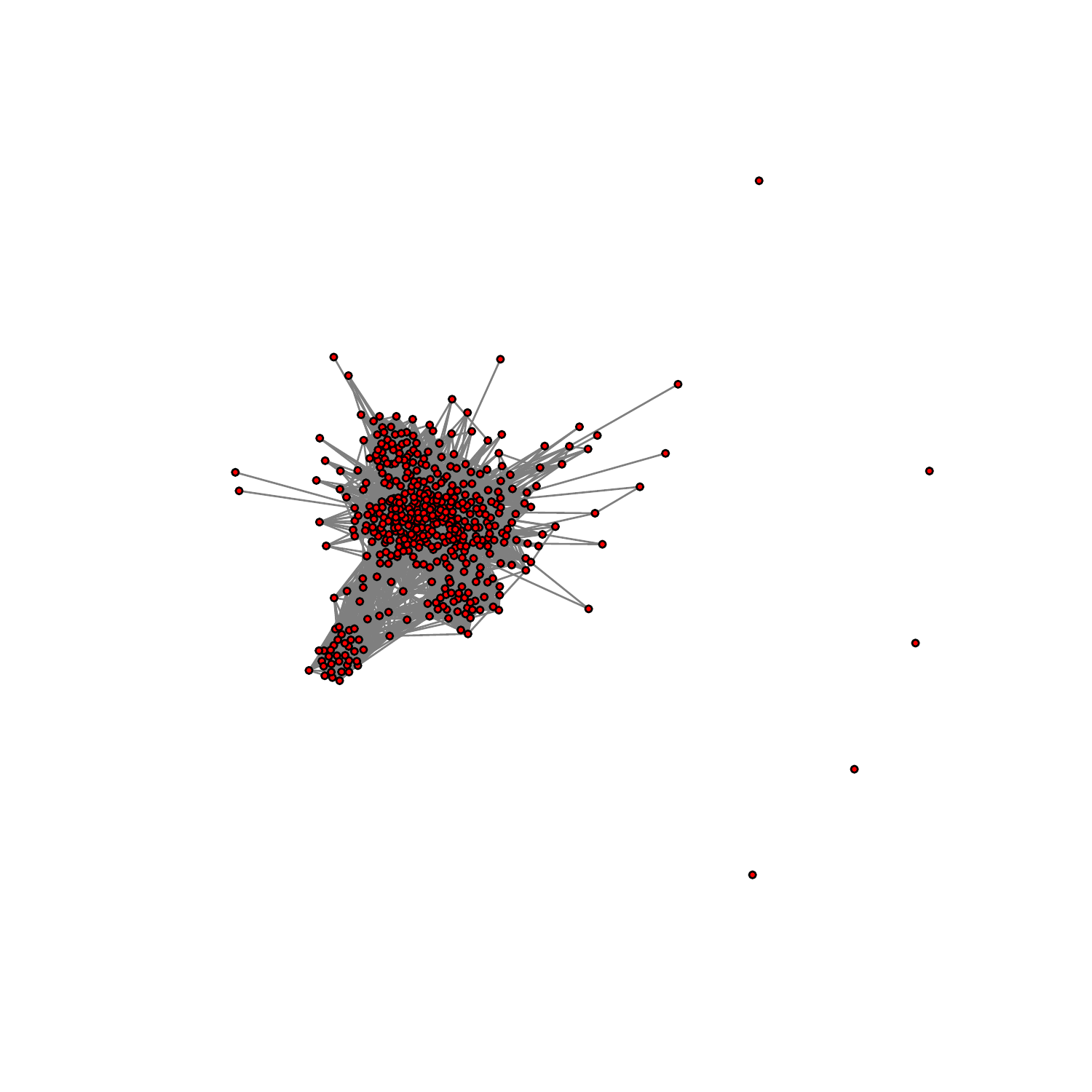} \\
(a) & (b) & (c) & (d)\\
\hline
\end{tabular}\end{minipage}}
\caption{Graphical models estimated from \texttt{stockdata}. Every node in the graph corresponds to one of the $p=452$ stocks.}
\label{Oellerer:tab:network}
\end{figure}
\medskip

Estimating a cellwise robust, sparse precision matrix is not only interesting in graphical models. As an example consider linear discriminant analysis, where each observation belongs to one of $K$ groups. The goal is then to assign a new observation $\mathbf{x}\in\mathbb{R}^p$ to one of those $K$ groups. Assuming a normal distributions $\mathcal{N}(\mathbf{\mu}_k,\bs \Sigma)$ for observations of group $k\in\{1,\ldots,K\}$, the Bayes optimal solution is found via the linear discriminant function 
\begin{align*}
\delta_k(\mathbf{x})=\mathbf{x}^\top\bs \Sigma^{-1}\mathbf{\mu}_k-\frac{1}{2}\mathbf{\mu}_k^\top\bs\Sigma^{-1}\mathbf{\mu}_k+\log \pi_k,
\end{align*}
where $\pi_k$ is the a priori probability of belonging to group $k$. Replacing $\Sigma^{-1}$ with the correlation based precision matrix estimated from the centered data (where each observation is centered by the coordinatewise median computed over the observations belonging to the same group) results in a cellwise robust estimator for high-dimensional linear discriminant analysis. The final estimate may not be sparse anymore, but it is very robust under cellwise contamination. Furthermore, it can be computed even if $p>n$.

Cellwise robust, sparse precision matrix estimation can also be used to obtain cellwise robust, sparse regression of $\mathbf{y}\in\mathbb{R}^n$ on $\bo X\in\mathbb{R}^{n\times p}$. Partitioning the joint sample covariance estimate of $(\bo X, \mathbf{y})$ and its inverse into 
\begin{align*}
\hat{\bs \Sigma}=\begin{pmatrix}
\hat{\bs \Sigma}_{\bo X\bo X} & \mathbf{\hat{\sigma}}_{\bo X \mathbf{y}}\\
\mathbf{\hat{\sigma}}_{\bo X \mathbf{y}}^\top & \hat{\sigma}_{\mathbf{y}\mathbf{y}}
\end{pmatrix}
\qquad
\hat{\bs \Theta}=\begin{pmatrix}
\hat{\bs \Theta}_{\bo X\bo X} & \mathbf{\hat{\theta}}_{\bo X\mathbf{y}}\\
\mathbf{\hat{\theta}}_{\bo X\mathbf{y}}^\top & \hat{\theta}_{\mathbf{y}\mathbf{y}}
\end{pmatrix}
\end{align*}
the least squares estimator can be rewritten as 
\begin{align*}
\mathbf{\hat{\beta}}_{LS}=\hat{\bs\Sigma}_{\bo X\bo X}^{-1}\hat{\mathbf{\sigma}}_{\bo X\mathbf{y}}=- \frac{1}{\theta_{yy}}\hat{\bs \Theta}_{\bo X \mathbf{y}}
\end{align*}
using the partitioned inverse formula \citep[][14.11]{Seber}. With the correlation based precision matrix estimate $\hat{\bs \Theta}_{\bo S}((\bo X, \mathbf{y}))$ computed jointly from $(\bo X, \mathbf{y})$, we obtain a cellwise robust, sparse regression estimate computable in high-dimensions
\begin{align*}
\mathbf{\hat{\beta}}=- \frac{1}{(\hat{\bs \Theta}_{\bo S}((\bo X, \mathbf{y})))_{p+1,p+1}}(\hat{\bs \Theta}_{\bo S}((\bo X, \mathbf{y})))_{1:p,p+1}.
\end{align*}

\section{Conclusions}
\label{Oellerer:sec:conclusion}

We have introduced a cellwise robust, correlation based precision matrix estimator. We put forward the following simple procedure: (i) compute the robust scale estimators $Q_n$ for each variable (ii) compute the robust correlation matrix from the normal scores, as in Equation~(\ref{Oellerer:eq:Gauss}) (iii) construct then the robust covariance matrix from these correlations and robust scale, as in Equation~(\ref{Oellerer:eq:pairwise}) (iv) use the latter as input for the GLASSO, returning $\hat{\bs \Theta}_{\bo S}(\bo X)$. It is formally shown that the proposed estimator features a very high breakdown point under cellwise contamination. As its definition is very simple, the estimator can be computed very fast, even in high-dimensions. 

The simulation results presented in Section~\ref{Oellerer:sec:sim} discuss the results of the various estimators \textit{including} their selection of the regularization parameter $\rho$. As can be seen from a small simulation study with $p=60$, kindly provided by a referee, the bad performance of `GlassoNPDQn' needs to be mainly attributed to the selection of $\rho$. When `GlassoNPDQn' is run with the regularization parameter estimated by `GlassoGaussQn', the two methods performed similar. This problem also occurred for clean data. Replacing CV by BIC did not help to improve `GlassoNPDQn': The performance in comparison to `GlassoGaussQn' was still similar as in Tables~\ref{Oellerer:tab:low} and \ref{Oellerer:tab:high}. Analyzing the reason for the bad performance of `GlassoNPDQn' with respect to the selection of the regularization parameter is left for future research.

Compared to the covariance based approach, a correlation based approach results in a simpler estimator. More importantly, it achieves a substantially higher breakdown point, is considerably faster to compute and yields more accurate estimates when the regularization parameter is selected using BIC or the new cross-validation criterion presented in Section~\ref{Oellerer:sec:rho}.

\subsection*{Acknowledgments}
The authors wish to acknowledge the helpful comments from two anonymous referees. The first author gratefully acknowledges support from the GOA/12/014 project of the Research Fund KU Leuven.

\bibliographystyle{plainnat}
\bibliography{sparsecov_Revision}

\end{document}